\documentclass{article}
\usepackage{arxiv}
\usepackage[OT2, T1]{fontenc}
\usepackage[russian, english]{babel}
\usepackage{caption}
\usepackage{lscape} 
\usepackage{amsmath}
\usepackage{geometry}
\usepackage{amsthm}
\usepackage{amsfonts}
\usepackage{mathtools}
\usepackage{verbatim}
\usepackage{mathtools}
\usepackage{comment}
\usepackage{verbatim, times, booktabs, dcolumn, setspace, fullpage}
\DeclarePairedDelimiter\abs{\lvert}{\rvert}
\newtheorem{theorem}{Theorem}

\newtheorem{lemma}{Lemma}
\newtheorem{chara}{Characterization}
\newtheorem{rem}{Remark}
\usepackage{graphicx}
\usepackage{wrapfig}
\usepackage{caption}
\usepackage{mathtools}
\usepackage{float}
\usepackage[caption = false]{subfig}

\usepackage[english,noautotitles-r]{SASnRdisplay} 
\usepackage{url}
\lstdefinestyle{r-output}{
style = r-style,
style = r-output-user,
}

\DeclareMathOperator\erf{erf}
\usepackage{tikz,xcolor}
\usepackage{hyperref}
\definecolor{lime}{HTML}{A6CE39}
\DeclareRobustCommand{\orcidicon}{%
	\begin{tikzpicture}
	\draw[lime, fill=lime] (0,0) 
	circle [radius=0.16] 
	node[white] {{\fontfamily{qag}\selectfont \tiny ID}};
	\draw[white, fill=white] (-0.0625,0.095) 
	circle [radius=0.007];
	\end{tikzpicture}
	\hspace{-2mm}
}

\foreach \x in {A, ..., Z}{%
	\expandafter\xdef\csname orcid\x\endcsname{\noexpand\href{https://orcid.org/\csname orcidauthor\x\endcsname}{\noexpand\orcidicon}}
}

\title{Characterization-based approach for construction of goodness-of-fit test for L\'evy distribution}
\author{\v Zikica Luki\'c\orcidA{}\\
University of Belgrade, Faculty of Mathematics\\
Belgrade, Serbia\\
\texttt{zikicamaster@gmail.com}
\And
Bojana Milo\v sevi\'c\orcidB{}\\
University of Belgrade, Faculty of Mathematics \\
Belgrade, Serbia\\
\texttt{bojana.milosevic@matf.bg.ac.rs}
}

\date{ }

\begin{document}

\maketitle

\begin{abstract}
    {{The} Lévy distribution, alongside  the Normal and Cauchy distribution, is one of { the} only three stable distributions whose density can be obtained in a closed form. However, there are only a few specific goodness-of-fit tests for the Lévy  distribution. In this paper two novel classes of goodness-of-fit tests for { the} Lévy  distribution are proposed. Both tests are based on $V$-empirical Laplace transforms. New tests are scale free under the null hypothesis, which makes them suitable for testing the composite hypothesis. The finite sample and limiting properties of test statistics are obtained. In addition, a generalization of the recent Bhati--Kattumannil goodness-of-fit test to { the} Lévy  distribution is considered. For assessing the quality of novel and competitor tests, the local Bahadur efficiencies are computed, and a wide power study is conducted. Both criteria clearly demonstrate the quality of the new tests. The applicability of the novel tests is demonstrated with two real-data examples.}
    \emph{\textbf{MSC Classification 2020}}  62E10, 62G10, 62G20

\end{abstract}

\section{Introduction}\label{sec:1}

The L\'evy distribution is one of the three stable distributions whose density has a closed form \cite{Zol}, given by  
{
\begin{align}
\label{1}
f(x; \lambda, \mu)=\sqrt{\dfrac{ \lambda}{2\pi}}\frac{e^{-\dfrac{ \lambda}{2(x-\mu)}}}{(x-\mu)^{\frac{3}{2}}},\;x\geq\mu,\; \lambda>0, \; \mu\in \mathbb{R}.
\end{align}

That property makes it especially  attractive in the scientific community, and consequently, it has many applications (see, e.g. \cite{Ebeling,West, rogers2008multiple, vinaya2018effect}). 
{ Therefore, it has been of huge importance to develop methods for parameter estimation as well as appropriate goodness-of-fit (GOF) tests.}

{ 
Maximum likelihood estimation of $\lambda$ when $\mu=0$ is covered in \cite{ali2005inference}, while  the case of both unknown parameters is addressed in \cite{achcar2018use}. However, the derivation is unclear and the numerical calculation of MLE yields different estimates.
It is worth mentioning that there are results about the MLE for the parameters of stable distribution in the general case. In \cite{nolan2001maximum} (see also \cite{tian2016parameter}),  one can read that this estimation procedure is computationally demanding and it requires the maximum searching procedure to be carefully implemented. 
The method proposed in \cite{mcculloch1986simple} avoids numerical optimization, but asymptotic properties seem to be unknown. The method proposed in \cite{koutrouvelis1980regression} produces consistent estimates, but asymptotic properties of the estimator seem to be unknown as well.
In practice, very often, the value of the location parameter can be deduced based on the nature of the phenomena, and therefore might be assumed as a known fixed value. Therefore in what follows, we assume that $\mu=0$ and the case of unknown $\mu$  is beyond the scope of this paper.



}
{ For the testing GOF to L\'evy distribution, one might use classical empirical distribution function (EDF)-based GOF tests, as in \cite{o1998note, lilliefors1967kolmogorov}. The usage of GOF tests for $\alpha$-stable distributions for $\alpha=0.5$ is possible as well \cite{pitera2022goodness}.} { However, as far as we know, the only specific GOF to L\'evy distribution, is proposed  in \cite{bhati2020jackknife}.}

One of the primary goals of this paper is to fill in the existing gap in the literature with proposal of two new scale-free classes of GOF tests. New tests belong to the group  of the equidistribution-characterization-based tests which recently, due to their nice properties, became very attractive. See, e.g., \cite{PUBL, milovsevic2016asymptotic,SORT}, for GOF tests for exponentiality, \cite{obradovic2015goodness, allison2021distribution} for GOF tests  for Pareto distribution, \cite{NIKITINLOGISTICKA} for GOF tests  for logistic distributions, etc.   


The common approach to assess the quality of tests is to find their power against different alternatives. This approach, for the GOF for the L\'evy distribution, is used in \cite{bhati2020jackknife}. Another approach, especially useful for a large sample comparison, is the notion of asymptotic efficiency.  Bahadur asymptotic efficiency or its approximation is shown to be an attractive option when dealing with tests with non-normal limiting properties (see, e.g., {\cite{nikitinKnjiga, nasArxiv, meintanis2022bahadur}}) { and is recently used as one of the main criteria for the evaluation of novel proposals (see, e.g., \cite{ragozin2021new,ebner2023test})}. Therefore,  we apply it for {the} novel and {the Bhati--Kattumanil tests}. We emphasize that this criterion has not been used before for accessing the quality of any GOF test to { the} L\'evy distribution.
}

This paper is organized in the following manner. In Section \ref{sec:2} we revisit {the test proposed in } \cite{bhati2020jackknife}.  Section \ref{sec:3} is dedicated to the proposal of new class{es} of goodness-of-fit tests and {their} limiting properties. The asymptotic efficiency of considered tests is derived in 
Section \ref{sec:4}, while the results of { the} empirical power study are presented in Section \ref{sec:5}.  An illustration of the usage of novel tests on real data is given in Section \ref{sec:6}. {All proofs are {included} in Appendix A. Appendix B contains  the generalization of Bhati--Kattumanil test, while used real-data sets and their visual representations are given in Appendix C.}
{ Appendix D contains the results related to the case of the median-based estimator, while Appendix E contains empirical percentiles of the null distribution of some representatives of the proposed classes of GOF test statistics. }

{ We will denote with $F(x;\lambda)$ the distribution function of {the} L\'evy distribution with scale parameter $\lambda$. { The} density of {the} L\'evy distribution with scale parameter $\lambda$ will be denoted with $f(x;\lambda)$.{The} standard L\'evy distribution is defined as {the} L\'evy distribution with scale parameter $\lambda=1$. 
}
In the sequel, the distribution function (df) and density of standard L\'evy distribution { will} be denoted with $F_0(x)$ and $f_0(x)$ for the reasons of brevity. 
\section{On Bhati--Kattumanil test statistic } \label{sec:2}

In \cite{AhsNev1}, Ahsanullah and Nevzorov proved the following  characterization of L\'evy distribution.

\begin{chara}  \label{karakterizacija}
Suppose that $X, Y$ and $Z$ are independent and identically distributed random variables with density $f$ defined on $(0, \infty)$. Then
$$Z\text{ and }\dfrac{aX + bY}{\big(\sqrt{a}+\sqrt{b}\big)^2}\text{, }0 < a, b < \infty$$
are identically distributed if and only if $f$ is a density of L\'evy distribution with arbitrary scale parameter $\lambda$.
\end{chara}
{The characterization is based upon the stability of the L\'evy distribution \cite{feller2008introduction}.} In  a  view of Characterization \ref{karakterizacija}, for $a=b=1$, Bhati and Kattumannil in  {\cite{bhati2020jackknife} proposed the test statistic
\begin{equation*}
    T_n^*=\int_{{\mathbb{R}^+}}\Big(\frac{1}{\binom{n}{2}}\sum\limits_{j<i}I\Big\{\frac{X_i+X_j}{4}\leq t\Big\}-F_n(t)\Big)dF_n(t),
\end{equation*}
which is an integrated difference between U-empirical df of $X$ and  $\frac{X+Y}{4}$.  This statistic is a hybrid U-statistic, asymptotically equivalent to the non-degenerate U-statistic
\begin{equation*}
    T_n=\frac{1}{\binom{n}{3}}\sum\limits_{k<j<i}I\Big\{\frac{X_i+X_j}{4}\leq X_k\Big\}-\frac{1}{2}.
\end{equation*}
} They showed that, under $H_0$, 
the distribution $\sqrt{n}T_n$  converged to a centred Gaussian distribution with variance equal to 
\begin{align}\label{sigmaG}
    \sigma^2_T=Var\Big(\int\limits_0^\infty2\big(1-F_0\Big(\frac{X+y}{4}\Big)\Big)f_0(y)dy+F_0(X)\big).
\end{align}

Since they were not being able to calculate asymptotic variance, they proceeded with considering empirical jackknife and jackknife-adjusted versions of $T_n$. The fact  that the test statistic is scale free under { the} null hypothesis, as it will be shown later, significantly simplifies testing procedures and the usage of the aforementioned resampling procedures might be skipped.

{
{Our} numerical calculation yields $\sigma^2_T= 0.0235051$. 

In the rest of this paper, we { will} denote the statistic $T_n$ with $\bar I^{[1, 1]}_n$, and its generalization for arbitrary $a,b>0$, { presented in Appendix B}, with $\bar I^{[a, b]}_n$. Since the properties of $\bar I^{[a, b]}_n$ aren't significantly different than those of the $a=b=1$, we present them in Appendix B.
}
\section{New classes of goodness-of-fit tests}\label{sec:3}
Equality in distribution of two random variables can also be expressed  through the equality of their Laplace transforms. 
The test statistic might be constructed as a function of the difference among corresponding U- or V-empirical Laplace transforms. This approach has been used for the first time in \cite{milovsevic2016new}, and was further explored in \cite{SORT, nasArxiv}.
{Taking into account the discussions  from mentioned papers, in a view of Characterization }  \ref{karakterizacija}, we propose new classes of test statistics $\mathcal{J}=\{J_{n,a},a>0\}$ and $\mathcal{R}=\{R_{n,a},a>0\}$, where
\begin{align}\label{statJn}
J_{n,a}&=\sup\limits_{t>0}\Big \vert \Big (\frac{1}{n^2}\sum\limits_{i, j} e^{-\frac{t(Y_{i}+Y_{j})}{4}}-\frac{1}{n}\sum\limits_{i} e^{-tY_{i}}\Big )e^{-at} t^{\frac{3}{2}}\Big \vert \\
&\nonumber=\sup\limits_{t\in [0, 1]}\Big \vert \Big (\frac{1}{n^2}\sum\limits_{i, j} t^{\frac{Y_{i}+Y_{j}}{4}}-\frac{1}{n}\sum\limits_{i} t^{Y_{i}}\Big )t^{a}\big (-\log t\big)^{\frac{3}{2}}\Big \vert ,\\
\label{statRn}
R_{n, a}& =\int_{\mathbb{R}^+}\Big (\frac{1}{n}\sum\limits_{i} e^{-tY_{i}}-\frac{1}{n^2}\sum\limits_{i, j} e^{-\frac{t(Y_{i}+Y_{j})}{4}}\Big )e^{-at} t^{\frac{3}{2}}dt\\
&\nonumber=\frac{3\sqrt{\pi}}{4n^2}\sum_{i, j}\Bigg(\frac{1}{\big(a+\frac{Y_i+Y_j}{4}\big)^\frac{5}{2}}-\frac{1}{2\big(a+Y_i\big)^\frac{5}{2}}-\frac{1}{2\big(a+Y_j\big)^\frac{5}{2}}\Bigg), 
\end{align}
and  $Y_k=\frac{X_k}{\hat{\lambda}}$  and $\hat{\lambda}$ is MLE of $\lambda$, given by
\begin{align}\label{lambdaMLE}\hat{\lambda}=\frac{n}{\sum\limits_{k=1}^n \frac{1}{X_k}}.\end{align}
{
{Note that} one might { also} opt for the median-based estimator (MBE):
\begin{equation*}
     \hat{\lambda}_{MBE} := 2(\erf^{-1}(1/2))^2 \widetilde{x},
 \end{equation*}
where  $ \widetilde{x}$ is the sample median, and where
\begin{equation*}
   \erf(x)=\frac{2}{\sqrt{\pi}}\int\limits_x^\infty e^{-t^2} dt
\end{equation*}
denotes the complementary error function. { The power study for this approach is presented in Appendix D. We will assume $\hat{\lambda}$ is the MLE in the rest of this paper.}}

Here the function {$e^{-at}t^{3/2}$} plays a role of weight function. 
Therefore, test statistics might be modified with the selection of another weight function.{ Since, under $H_0$, the values of test statistics should be small, we take large values of  $J_{n,a}$  and $\vert R_{n,a}\vert $   to be significant.} 

In the next  two theorems, we present limiting distributions of $\sqrt{n}J_{n,a}$ and $\sqrt{n}R_{n,a}$ under $H_0$.

\begin{theorem}\label{asimptotikaJ}
Let $a\geq 1$ and $X_1, X_2, \dots, X_n$ be i.i.d random variables  distributed according to the L\'evy law with scale parameter $\lambda$. Then the following holds:
\begin{equation*}
    \sqrt{n} J_{n, a}\overset{{D}}{\to} \sup\limits_{t\in [0, 1]}\mid \xi (t)\mid,
\end{equation*}
where $\xi(t)$ is a centred Gaussian random process, having the following covariance function:

\begin{align*}
    K(s,t)&=
  s^a t^a (-\log (s))^{3/2} (-\log (t))^{3/2} \Big(-e^{-\sqrt{2} \big(\sqrt{-\log (s)}+\sqrt{(-\log (t))}\big)}-2 e^{- \sqrt{-2(\log(s)-\frac14\log(t))}+\sqrt{-\frac{\log (t)}{2}}}\\&-2 e^{- \sqrt{2(-\log (t)-{\frac14}\log(s))}+\sqrt{-\frac{\log (s)}{2}}}+4 e^{-\frac{\sqrt{-\log (s t)}+\sqrt{-\log (s)}+\sqrt{-\log (t)}}{\sqrt{2}}}+e^{ \sqrt{-2\log (s t)}}\Big).
\end{align*}
\end{theorem}
{ Empirical 95th percentiles of $\sqrt{n}J_{n, a}$, presented in Table \ref{critvalJ05} and Table \ref{critvalJ5} in Appendix E are in concordance with Theorem \ref{1}.}
\bigskip

\begin{rem}

Let $X_1,X_2...,X_n$ be i.i.d. random variables distributed as $X\in F_A$ and let $F_A$ be a fixed alternative distribution such that $E\big(\frac{1}{X}\big)^2<\infty$. Denote with $\zeta=\frac{1}{E\big(\frac{1}{X}\big)}$. The law of large numbers along with the continuous mapping theorem gives us $\hat{\lambda}=\frac{n}{\sum\limits_{i=1}^n \frac{1}{X_i}}\xrightarrow[n\to\infty]{P}\zeta$.
Next, applying the law of large numbers for U- and V-statistics with estimated parameters \cite{iverson1989effects}, and similar arguments as in the proof of Lemma \ref{lemma1} we have that
$J_n\overset{P}{\to}\sup_{t\in[0,1]}|(E(t^{\frac{X_1}{\zeta}})-E(t^{\frac{X_1+X_2}{4\zeta}}))t^a(-\log t)^{\frac{3}{2}}|$ which is equal to 0 if and only if $\frac{X_1}{\zeta}$ and $\frac{X_1+X_2}{4\zeta}$ are equally distributed, i.e. iff the null hypothesis holds. From this, consistency of $J_{n,a}$ follows.

\end{rem}
\begin{theorem}\label{asimptotikaR}
Let $a\geq 1$ and $X_1, X_2, \dots, X_n$ be i.i.d random variables  distributed according to the L\'evy law with scale parameter $\lambda$. Then, for every $a>0$,
the asymptotic distribution of $\sqrt{n}R_{n, a}$ as $n\to\infty$ is normal $\mathcal{N}(0, \sigma^2_R(a))$ where $\sigma^2_R(a)= 4E\zeta(X;  a)^2$ and  $\zeta(x;a)$ is defined in \eqref{zeta}.
\end{theorem}
The values of $\sigma^2_{R}(a)$ cannot be calculated analytically. However, it is possible to calculate them numerically. Some values of $\sigma^2_{R}(a)$ are presented in Table \ref{tab: dispR}. Therefore, the testing procedure can also be done using {the} standardized test statistic
\begin{align*}
     \widetilde{R}_{n,a}=\sqrt{n}\frac{ R_{n,a}}{\sigma_R(a)}
\end{align*}
which, for large samples, under $H_0$, can be approximated with standard normal distribution.

\begin{table}[htbp]
\caption{Values of $\sigma^2_R(a)$ for different values of $a$.}
\label{tab: dispR}
\centering
\begin{tabular}{@{}llllll@{}}
\toprule
$\sigma^2_R(0.2)$ & $\sigma^2_R(0.5)$ & $\sigma^2_R(1)$ & $\sigma^2_R(2)$ & $\sigma^2_R(5)$ \\ \midrule
4.58804 & 0.2672024 & 0.02068868 & 0.001194688 & $1.925016\cdot 10^{-5}$ \\ \bottomrule
\end{tabular}
\end{table}
{ Empirical 95th percentiles of $\sqrt{n}|R_{n, a}|$, presented in Table \ref{critvalL.5} and  Table \ref{critval5}  in Appendix E, are in concordance with Theorem \ref{asimptotikaR}.}

\section{Asymptotic efficiency}\label{sec:4}

In recent times, the Bahadur efficiency has become a very popular tool for  stochastic comparison of test performance in {the} large sample case.   
In this section we make a brief review of Bahadur theory. For { more} details, we refer to \cite{nikitinKnjiga}.

Let $\mathcal{G}=\{g(x;\theta),\;\theta>0\}$ be a family of alternatives
{density} functions, such that $g(x;0)$ has the L\'evy distribution with arbitrary scale parameter,  and { $\int_{
\mathbb{R}^+
}\frac{1}{x^2}g(x;\theta)<\infty$ for $\theta$ in the neighbourhood of 0, and  some additional regularity
conditions for U-statistics with non-degenerate kernels hold \cite{nikitinMetron, meintanis2022bahadur}. Let also $\{T_n\}$ and $\{V_n\}$ be two sequences of test statistic that we want to compare.

Then for any alternative distribution from $\mathcal{G}$ the relative Bahadur efficiency of the $\{T_n\}$ with respect to $\{V_n\}$ can be expressed as 
\begin{align*}
    e_{(T,V)}(\theta)=\frac{c_T(\theta)}{c_V(\theta)},
\end{align*}
where $c_{T}(\theta)$  and $c_V(\theta)$ are  the Bahadur exact slopes, functions proportional
to the exponential rate of decrease of each test size when the sample size
increases.  It is usually assumed that $\theta$ belongs to the neighbourhood of 0, and in such cases, we refer to {the} local relative Bahadur efficiency of considered sequences of test statistics.

It is well known that for the Bahadur slope function Bahadur--Ragavacharri inequality holds \cite{raghavachari1970theorem}, that is
\begin{align*}
    c_T(\theta)\leq 2K(\theta),
\end{align*}
where
$K(\theta)$ is the minimal Kullback--Leibler distance
from the alternative to the class of null hypotheses, i.e. in the case of our null hypothesis
\begin{align*}
    K(\theta)=\inf_{\lambda>0}K(\theta; \lambda)=\inf_{\lambda>0}\int_{\mathbb{R}^+}\log\Big(\frac{g(x;\theta)}{f(x; \lambda)}\Big)g(x;\theta)dx.
\end{align*}
This justifies the definition of the local absolute Bahadur efficiency by
\begin{align}\label{effT}
    eff(T)=\lim_{\theta\to 0}\frac{c_T(\theta)}{2K(\theta)}.
\end{align}

If   the
sequence $\{T_{n}\}$ of test statistics under the alternative converges
in probability to some finite function $b(\theta)>0$ and 
the  limit 
\[\label{ldf}
\lim_{n\leftarrow\infty}n^{-1}\log P_{H_{0}}(T_{n}\geq t)=-f_{LD}(t)
\]
exists for any $t$ in an open interval $I$, on which $f_{LD}$ is continuous
and $\{b(\theta),\theta>0\}\subset I$ then   the Bahadur exact slope
is  equal to
\begin{equation}\label{slope}
c_{T}(\theta)=2f_{LD}(b(\theta)).
\end{equation}
However, in many cases, the calculation of the large deviation function, and consequently the Bahadur slope, turns out to be almost an insurmountable obstacle.  

If the function \eqref{ldf} cannot be calculated for $\theta$ approaching zero, instead of calculating Bahadur slope we could calculate approximate Bahadur slope $c_T^*(\theta)$ { which usually} locally coincides with the exact one. For { the} calculation of the approximate slope, we do not need the tail behaviour of d.f. of statistics $T_n$ but the tail behaviour of its limiting distribution, which is often easier to calculate.
In particular, if the limiting distribution function of $T_n$, under $H_0$, is $F_T$, whose tail behaviour is given by $\log (1-F_T(t)) =-\frac{a^*_T t^2}{2}(1+o(1)),$ where $a_T$ is the positive real number and $o(1)\to 0$ when $t\to\infty$, and  the limit in
probability of $\frac{T_n}{\sqrt{n}}$ is $b^*_T(\theta)>0$, then  $$c_T^*(\theta)=a_T^*(b^*_T(\theta))^2.$$

In addition, the local (approximate) slope of likelihood ratio tests is equal to $2K(\theta)$ \cite{bahadur1967rates}, therefore it is reasonable to approximate  \eqref{effT}  by replacing $c_T(\theta)$ with $c_T^*(\theta).$

In the next theorem{,} we provide the behaviour of $2K(\theta)$ when $\theta$ approaches zero.
\begin{theorem}\label{KL}
For a given alternative density $g(x; \theta)$ whose distribution belongs
to $\mathcal{G}$, such that $g(x;0)$ is given by \eqref{1}, it holds
\begin{align*}
   2K(\theta)=\Bigg( \sqrt{\frac{2\pi}{\lambda}}\int\limits_{\mathbb{R^+}} (g'(x;\theta))^2e^{-\frac{\lambda}{2x}}x^{-\frac{3}{2}}dx-\frac{\lambda}{2}\Big(\int\limits_{\mathbb{R^+}}\frac{g'(x;\theta)}{x}dx\Big)^2\Bigg)\cdot\theta^2+o(\theta^2),\;\theta\to 0.
\end{align*}
\end{theorem}

In the following theorems, we present Bahadur exact and approximate slopes  of  Bhati--Kattumanil statistic and our  statistics, respectively.

\begin{theorem}\label{Tint} For an alternative $g(x;\theta)$ from
$\mathcal{G}$, the Bahadur exact slope of the statistic $I_{n}^{[1,1]}$
is 
\begin{align*}
c_{I}(\theta)=\frac{\Big({ 3}\int_{\mathbb{R}^+}\varphi(x)g'_{\theta}(x; 0)dx\Big)^{2}}{\sigma_T^{2}}\cdot\theta^{2}+o(\theta^{2}),\theta\to 0,
\end{align*}
where
\begin{equation*}
\varphi(x) =\int\limits_0^\infty2\big(1-F\Big(\frac{x+y}{4}; \lambda\Big)\big)f(y;\lambda)dy+F(x;\lambda).
\end{equation*}
\end{theorem}
\begin{proof}[Proof of Theorem \ref{Tint}.]
The proof follows directly from \cite[Theorem 3]{nikitinMetron}. 
\end{proof}

\begin{theorem}\label{Tsup}
For an alternative $g(x;\theta)$ from
$\mathcal{G}$, the Bahadur approximate  slope of the statistic $J_{n, a}$ is equal to
\begin{align}\label{Bslope}
    c_{J}(\theta)=\frac{\sup_{t\in [0, 1]}\Big({2}\vert \int_{\mathbb{R}^+}\psi(x; t, a)g'_{\theta}(x; 0)dx\vert \Big)^{2}}{\sup_{t\in [0, 1]}\sigma^2(t)}\cdot \theta^2+o(\theta^2), 
\end{align}
where $\sigma^2(t)=\sup_{t\in[0, 1]} K(t, t)$, and
\begin{equation*}
     \psi(x; t, a)=\frac{1}{2} t^a (-\log (t))^{3/2} \Big(-2 t^{\frac{x}{4}} e^{-\frac{\sqrt{-\log (t)}}{\sqrt{2}}}+t^x+e^{-\sqrt{2} \sqrt{-\log (t)}}\Big).
\end{equation*}
\end{theorem}

{\begin{theorem}\label{Rint} For an alternative $g(x;\theta)$ from
$\mathcal{G}$, the Bahadur exact slope of the statistic $R_{n,a}$
is 
\begin{align*}
c_{R}(\theta)=\frac{\Big({ 2}\int_{\mathbb{R}^+}\zeta(x)g'_{\theta}(x; 0)dx\Big)^{2}}{\sigma_R^{2}(a)}\cdot\theta^{2}+o(\theta^{2}),\theta\to 0,
\end{align*}
where $\zeta$ is the first projection of the symmetric kernel $Z$ given by
\begin{align*}
    \zeta(x; a)=E(Z(X_1, X_2;a)\vert X_1=x).
\end{align*}
The expression for $\zeta$ is cumbersome and the exact form can be found in Appendix, (see \ref{zeta}).
\end{theorem}

}

We consider {the} following classes of alternatives that belong to family $\mathcal{G}$: 
\begin{itemize}
   \item  a mixture of the standard  L\'evy distribution and the L\'evy distribution with scale  parameter $\lambda\neq 1$,  with density:
    \begin{equation*} 
       g_1^{[\lambda]}(x;\theta)=(1-\theta)f_0(x)+ \frac{\theta}{\lambda}f_0\big(\frac{x}{\lambda}\big),\; x>0,\;\theta\in (0,1);
    \end{equation*}
    \item a Lehmann alternative with density
    
    \begin{align*}
        g_2(x;\theta)=({1+\theta}) F_0(x)^{\theta}f_0(x),\; x>0,\;\theta>0;
    \end{align*}
    \item a contamination alternative with $g_2$, and parameter $\beta$   with density
   \begin{align*}
        g_3^{[\beta]}(x;\theta)=(1-\theta)f_0(x)+&\theta \beta F_0^{\beta-1}(x)f_0(x),\;x>0,\;\theta\in (0,1), \;\beta>0;
    \end{align*}    
  \item a first Ley--Paindaveine alternative \cite{ley2009cam} with density
    \begin{equation*}
       g_4(x;\theta) = (1+\theta F_0(x))f_0(x)e^{-\theta(1-F_0(x))},\; x>0,\;\theta>0;
    \end{equation*}
    \item  a second Ley--Paindaveine alternative \cite{ley2009cam} with density
    \begin{equation*}
        g_5(x;\theta) = f_0(x)(1-\theta\pi\cos\big(\pi F_0(x)\big)),\; x>0,\; \theta\in [0, \pi^{-1}]. 
    \end{equation*}
\end{itemize} 
In what follows we present a calculation of the local approximate Bahadur efficiency of $J_1$ and alternative $g_2(x;\theta)$, while  results for all  considered statistics and alternatives   are presented in Table \ref{LABRE}.

From Theorem \ref{KL} we obtain 
\begin{align*}
2K(\theta)&=\Bigg(\int\limits_\mathbb{R^+}\frac{e^{-\frac{1}{2x}} \big(\log \big(\erf \big(\frac{1}{\sqrt{2} \sqrt{x}}\big)\big)+1\big)^2}{\sqrt{2\pi } x^{3/2}}dx-\frac{1}{2}\int\limits_\mathbb{R^+}\frac{e^{-\frac{1}{2 x}} \big(\log \big(\erf  \big(\frac{1}{\sqrt{2} \sqrt{x}}\big)\big)+1\big)}{\sqrt{2 \pi } x^{5/2}}dx\Bigg)\cdot\theta^2+o(\theta^2)\\
&=0.0233005\theta^2+o(\theta^2),\;\theta\to0.
\end{align*}
Next, from Theorem \ref{Tsup}, it follows that
\begin{align*}
c_J(\theta)&=\frac{{ 4}\sup_{t\in[0, 1]} A(t)}{\sup_{t\in [0, 1]}\sigma^2(t)}\theta^2+o(\theta^2)=\frac{{4}\sup_{t\in [0, 1]}\Big(\vert \int_{\mathbb{R}^+}\psi(x; t, a)g'_{\theta}(x; 0)dx\vert \Big)^{2}}{\sup_{t\in [0, 1]}\sigma^2(t)}\cdot \theta^2+o(\theta^2).    
\end{align*} 
We have that
\begin{align*}
   & \sup\limits_{t\in[0, 1]} A(t)=  \sup\limits_{t\in[0, 1]}\Big(\int\limits_\mathbb{R^+}t^a (-\log (t))^{3/2} \big(\log \big(\erf\Big(\frac{1}{\sqrt{2} \sqrt{y}}\Big)\big)+1\big)\times\big.\\
    &\big.\frac{e^{-\sqrt{2} \sqrt{-\log (t)}-\frac{1}{2 y}} \big(-2 t^{y/4} e^{\frac{\sqrt{-\log (t)}}{\sqrt{2}}}+t^y e^{\sqrt{2} \sqrt{-\log (t)}}+1\big)}{2 \sqrt{2 \pi } y^{3/2}}dy\Big)^2\\
    &\approx 0.0000149667,
\end{align*}
We highlight that the maximum of { the} function $A(t)$ (presented in Figure \ref{fig: A}) is calculated numerically.

\begin{figure}[ht]
\centering
\includegraphics[width=8cm]{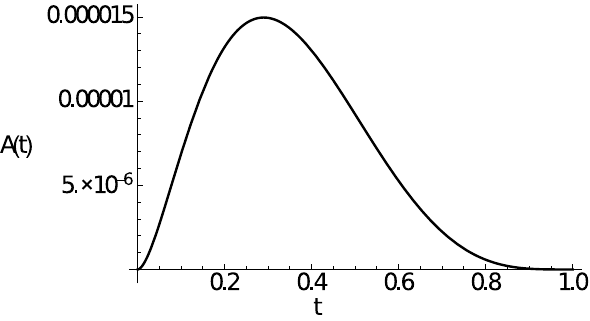}
\caption{Plot of the A$(t)$.}
\label{fig: A}

\end{figure}

Further,  $\sup_{t\in [0, 1]}\sigma^2(t)$ becomes
\begin{align*}
    \sup_{t\in [0, 1]} &\Big(\!-\! t^2 \big(-4 e^{-\frac{2 \sqrt{-\log(t^{5/4})}\!+\!\sqrt{-\log (t)}}{\sqrt{2}}}\!+\!4 e^{-\frac{\sqrt{-\log(t^2)}+2 \sqrt{-\log (t)}}{\sqrt{2}}}  
   \!+\!e^{-\sqrt{-2\log(t^2)}}
  \! -\!e^{-2  \sqrt{-2\log (t)}}\Big) \log ^3(t)\\&\approx {0.00388889} .
\end{align*}
and the calculation yields $eff(J_1)\approx 0.66$. The function $\sigma^2(t)$ is presented in Figure \ref{fig: sigma}.
\begin{figure}[ht]
\centering
\includegraphics[width=8cm]{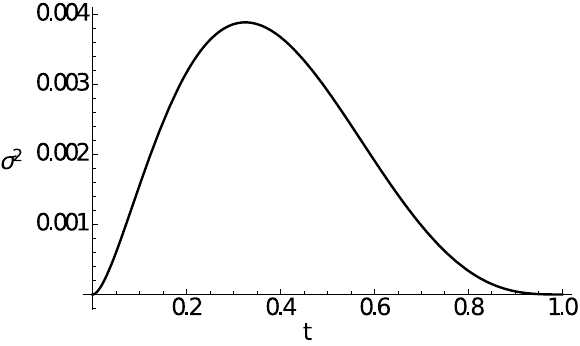}
\caption{Plot of $\sigma^2(t)$.}
\label{fig: sigma}

\end{figure}

\medskip

The results of our study are {displayed} in Table \ref{LABRE}. The tuning parameter $a$ significantly affects the efficiency of $R_{n, a}$ and $J_{n, a}$. In all considered cases, the Bahadur efficiency of $J_{n,a}$ is a decreasing function of $a$. This is not the case for the statistic $R_{n, a}$. However, it is notable that the maximal efficiency is attained in the neighbourhood of $a=1$. {The values of the local approximate Bahadur relative efficiencies for the generalized Bhati--Kattumanil statistic are given in Appendix B. }
{We can see that the new statistics outperform the Bhati--Kattumanil one, and that the statistic $R_{n, a}$ dominates the other two in terms of local approximate Bahadur efficiencies.}
\begin{rem}
It can be shown that for the test statistic $I^{[1, 1]}$ considered in this paper approximate and exact slopes coincide locally, and consequently, the approximate and exact Bahadur relative efficiencies w.r.t. LR test locally coincide because the approximate Bahadur slope for LR test equals to $2K(\theta)$, as we have noted earlier.
\end{rem}
\begin{table}[htbp]
\centering
\caption{Local approximate Bahadur relative efficiencies of $I^{[1, 1]}$,  $J_a$ and $R_a$ with respect to LR test}
\label{LABRE}
\centering
\begin{tabular}{llllll} 
\toprule
     & $g_1^{[10]}$   & $g_2$   & $g_3^{[3]}$    & $g_4$  & $g_5$    \\ 
\hline
$I^{[1, 1]}$  & 0.59   & 0.54   & 0.73   & 0.53  & 0.41    \\
$J_1$ & 0.91   & 0.66   & 0.79   & 0.68  & 0.69    \\
$J_2$ & 0.81   & 0.54   & 0.71   & 0.54  & 0.49    \\
$J_5$ & 0.56   & 0.36   & 0.52   & 0.35  & 0.25    \\
$J_{10}$        & 0.35   & 0.24   & 0.37   & 0.23  & 0.13    \\
$R_{0.2}$       & 0.53   & 0.79   & 0.61   & 0.80  & 0.86    \\
$R_{0.5}$       & 0.80   & \textbf{0.86} & 0.82   & \textbf{0.90} & \textbf{0.97}  \\
$R_{1}$         & \textbf{0.94} & 0.81   & \textbf{0.89} & 0.84  & 0.87    \\
$R_2$ & 0.93   & 0.69   & 0.84   & 0.70  & 0.65    \\
$R_5$ & 0.70   & 0.48   & 0.66   & 0.46  & 0.35    \\
\bottomrule
\end{tabular}

\end{table}

\section{Power  study}\label{sec:5}

In this section,  we explore finite sample properties of considered test statistics. In particular, we estimate the power of the tests, when the level of significance is $\alpha=0.05$, using Monte Carlo method with $N=10000$ replications.  The goal of this section is  to compare JEL and AJEL approaches from \cite{bhati2020jackknife} to the classical approach and to determine the empirical powers of new tests. The $p$-values are obtained utilizing the Monte Carlo approach.

{ The supremum in the calculation of the $J_{n, a}$ is obtained using grid search on 1000 equidistant points on [0, 1].

}
{ Other tests considered in the power study are:
\begin{itemize}
\item Lilliefors-corrected Kolmogorov--Smirnov (KS) test statistic: 
\begin{align*}KS_n=\sup\limits_{x\in \mathbb{R}^+}\abs{F_n(x)-F(x;\hat{\lambda})};\end{align*}
\item Lilliefors-corrected Cramer--von Mises (CVM) test statistic:
\begin{align*}
CVM^2_n=\int\limits_{\mathbb{R}^+}(F_n(x)-F(x;\hat{\lambda}))^2dF(x;\hat{\lambda});\end{align*}
\item Lilliefors-corrected Anderson--Darling (AD) test statistic: 
\begin{align*}AD^2_n=n\int\limits_{\mathbb{R}^+}\dfrac{(F_n(x)-F(x;\hat{\lambda}))^2}{F(x;\hat{\lambda})(1-F(x;\hat{\lambda}))}dF(x;\hat{\lambda});\end{align*}
\item $N_1^a$ test statistic, considered in \cite{pitera2022goodness}:
\begin{align*}
    N_1^a=\sqrt{n}\cdot\frac{\hat\sigma^2_x(5\%, 25\%)-\hat\sigma^2_x(75\%, 95\%)}{\hat\sigma^2_x(5\%, 95\%)};
\end{align*}

\item $N_1^b$ test statistic, considered in \cite{pitera2022goodness}:
\begin{align*}
    N_1^b=\sqrt{n}\cdot\frac{2.00\cdot\hat\sigma^2_x(5\%, 25\%)-1.01\cdot\hat\sigma^2_x(75\%, 95\%)}{\hat\sigma^2_x(5\%, 95\%)},
\end{align*}
\end{itemize}
where $\hat\sigma^2_x$ denotes the sample quantile conditional variance estimator \cite{pitera2022goodness}:
 \begin{align*}
     \hat\sigma_x^2(a, b)=\frac{1}{[nb]-[na]}\sum\limits_{i=[na]+1}^{[nb]} \Big(X_{(i)}-\hat\mu_X(a, b)\Big)^2,
 \end{align*}
and $X_{(i)}$ is the ith-order statistic, and  $\hat\mu_X(a, b)=\frac{1}{[nb]-[na]}\sum_{i=[na]+1}^{[nb]} X_{(i)}$ is the conditional sample mean.
}

{ In the case of Lilliefors-corrected classical tests large values are taken to be significant, while in the case of both small and large values of $N_1^a$ and $N_1^b$, the null hypothesis is rejected.}

We consider the following classes of alternative distributions:

\begin{itemize}
\item  Burr distribution, denoted by \text{Burr(a, b, c}),  with a density 
$$g_B(x;a,b,c ) = c b  \dfrac{\big (\frac{x}{a} \big )^{b-1}}{a  \big (1+ \big (\frac{x}{a} \big )^b \big )^{c+1}}, \; x>0,\; a>0,\; b>0,\; c>0 ;$$
\item Chen distribution, denoted by \text{Chen($\nu$, $\lambda$),} with a density
$$
g_C(x;  \nu, \lambda)=\nu\lambda x^{\lambda-1}e^{\nu(1-e^{x^\lambda})+x^\lambda}, \; x>0, \;\nu>0, \;\lambda>0;
$$
{
\item Fr\'echet distribution, denoted by \text{FR($a$, $b$)}, with a density
$$g_{FR}(x;a, b ) =  \frac{a}{b}\left(\frac{x}{b}\right)^{-(a+1)}\exp\left(-\left(\frac{x}{b}\right)^{-a}\right), \; x>0,\; a>0,\; b>0;$$
\item Gamma distribution, denoted by \text{$\Gamma$($a$, $b$)}, with a density
$$g_{\Gamma}(x;a, b ) =  \frac{x^{a-1}b^a e^{-bx}}{\Gamma(a)}, \; x>0,\; a>0,\; b>0;$$
\item log-logistic distribution, denoted by \text{LL($a$, $b$}), with a density
$$g_{LL}(x;a, b ) =  \frac{a (\frac{x}{b})^a}{x[1+a (\frac{x}{b})^a]^2}, \; x>0,\; a>0,\; b>0;$$}
\item log-normal distribution, denoted by \text{LN(a, b}), with a density
$$g_{LN}(x;a, b ) = \frac{1} {\sqrt{ (2 \pi )} b x }e^{\frac{- ( (\log x - a )^2} { 2 b^2 }}, \; x>0,\; a\in\mathbb{R},\; b>0;$$ 
\item $\chi^2$ distribution, denoted by $\chi^2_{n}$, with a density
$$g_{\chi^2}(x; n ) = \frac{1}{2^{\frac{n}{2}} \Gamma \big (\frac{n}{2} \big ) } x^{\frac{n}{2}-1} e^{-\frac{x}{2}}, \; x>0, \; n\in\mathbb{N};$$
\item half-normal distribution, denoted by \text{HN(a, b}), with a density
$$g_{HN}(x;a,b )=\frac{e^{-\frac{ (x-a )^2}{2b^2}}}{\sqrt{2\pi b^2}}+\frac{e^{-\frac{ (x+a )^2}{2b^2}}}{\sqrt{2\pi b^2}}, \; x>a,\; a\in\mathbb{R},\; b>0;$$
{\item shifted log-gamma distribution, denoted by \text{LG($a$, $b$}), with a density
$$g_{LG}(x;a, b ) = \frac{b^a}{\Gamma(a)}\frac{(\log (x+1))^{a-1}}{(x+1)^{b+1}}, \; x>0,\; a>0,\; b>0;$$} 
\item Weibull distribution, denoted by  \text{W(a, b}) with a density
$$g_W (x;  a, b )= \frac{a}{b}  \Big (\frac{x}{b} \Big )^{a-1} e^{(-\frac{x}{b})^a}, \; x>0,\; a>0,\; b>0;$$

\end{itemize} 
{The majority of these alternatives are considered in \cite{bhati2020jackknife}.}

\begin{landscape}
\begin{table}[htbp] \centering
\caption{Comparison of empirical powers - { ML estimate}}
\label{zdruzene_tabele}
\begin{tabular}{lllllllllllllllllll}
\hline
Distribution & n & $\bar I^{[ 1, 1]}$ & $J_{1}$ & $J_{2}$ & $J_{5}$ & $J_{10}$ & $R_{0.2}$ & $R_{0.5}$ & $R_{1}$ & $R_{2}$ & $R_{5}$ & JEL & AJEL & KS & CVM & AD & $N_1^a$ & $N_1^b$ \\ \hline
L\'evy(0, 0.5) & 25 & 0.05 & 0.05 & 0.04 & 0.05 & 0.05 & 0.05 & 0.05 & 0.05 & 0.05 & 0.05 & 0.06 & 0.04 & 0.06 & 0.05 & 0.05 & 0.05 & 0.05 \\
L\'evy(0, 1) & 25 & 0.04 & 0.05 & 0.05 & 0.05 & 0.05 & 0.05 & 0.05 & 0.05 & 0.05 & 0.05 & 0.05 & 0.04 & 0.06 & 0.05 & 0.05 & 0.06 & 0.05 \\
L\'evy(0, 2) & 25 & 0.05 & 0.05 & 0.05 & 0.05 & 0.05 & 0.05 & 0.05 & 0.05 & 0.05 & 0.05 & 0.05 & 0.04 & 0.05 & 0.05 & 0.05 & 0.05 & 0.05 \\
Burr(1.5, 0.5, 0.5) & 25 & 0 & 0.55 & 0.43 & 0.43 & 0.53 & 0.98 & 0.98 & 0.64 & 0.60 & 0.59 & 0.33 & 0.21 & \textbf{1} & \textbf{1} & \textbf{1} & 0.18 & 0.42 \\
LN(0, 1) & 25 & \textbf{0.99} & 0.71 & 0.89 & 0.98 & \textbf{0.99} & 0.37 & 0.63 & 0.82 & 0.95 & \textbf{0.99} & 0.29 & 0.24 & 0.86 & 0.88 & 0.91 & 0.30 & 0.38 \\
$\chi^2(3)$ & 25 & \textbf{1} & 0.61 & 0.83 & 0.96 & 0.99 & 0.36 & 0.48 & 0.69 & 0.87 & 0.97 & 0.52 & 0.39 & 0.96 & 0.96 & 0.98 & 0.57 & 0.59 \\
HN(0,1) & 25 & \textbf{1} & 0.41 & 0.60 & 0.83 & 0.92 & 0.43 & 0.35 & 0.45 & 0.65 & 0.85 & 0.53 & 0.42 & 0.96 & 0.95 & 0.97 & 0.73 & 0.63 \\
$\Gamma (3, 2)$ & 25 & \textbf{1} & \textbf{1} & \textbf{1} & \textbf{1} & \textbf{1} & 0.87 & 0.98 & \textbf{1} & \textbf{1} & \textbf{1} & 0.83 & 0.76 & \textbf{1} & \textbf{1} & \textbf{1} & 0.73 & 0.79 \\
W(2, 1) & 25 & \textbf{1} & 0.98 & \textbf{1} & \textbf{1} & \textbf{1} & 0.76 & 0.93 & 0.98 & \textbf{1} & \textbf{1} & 0.79 & 0.77 & \textbf{1} & \textbf{1} & \textbf{1} & 0.86 & 0.87 \\
$\Gamma(0.4, 2)$ & 25 & 0.16 & 0.44 & 0.45 & 0.44 & 0.44 & 0.98 & 0.97 & 0.56 & 0.48 & 0.43 & 0.06 & 0.04 & \textbf{0.99} & \textbf{0.99} & \textbf{0.99} & 0.22 & 0.03 \\
$W(0.4, 2)$ & 25 & 0 & 0.47 & 0.47 & 0.47 & 0.47 & 0.98 & 0.97 & 0.56 & 0.49 & 0.42 & 0.07 & 0.05 & 0.98 & \textbf{0.99} & \textbf{0.99} & 0.06 & 0.04 \\
$LN(0, 2)$ & 25 & 0.06 & 0.22 & 0.22 & 0.22 & 0.22 & \textbf{0.58} & 0.39 & 0.20 & 0.10 & 0.05 & 0.06 & 0.04 & 0.53 & 0.56 & 0.55 & 0.07 & 0.04 \\
$Chen(2, 0.4)$ & 25 & 0.06 & 0.46 & 0.46 & 0.47 & 0.46 & \textbf{0.93} & 0.85 & 0.47 & 0.33 & 0.22 & 0.07 & 0.05 & 0.91 & \textbf{0.93} & 0.92 & 0.24 & 0.05 \\
LG(7, 2) & 25 & \textbf{0.56} & 0.25 & 0.35 & 0.47 & 0.52 & 0.14 & 0.23 & 0.32 & 0.45 & 0.53 & 0.08 & 0.07 & 0.27 & 0.30 & 0.32 & 0.15 & 0.22 \\
LL(1, 2) & 25 & 0.24 & 0.13 & 0.07 & 0.10 & 0.14 & \textbf{0.51} & 0.28 & 0.11 & 0.08 & 0.11 & 0.06 & 0.05 & 0.47 & 0.48 & 0.48 & 0.09 & 0.10 \\
FR(1, 1) & 25 & 0.76 & 0.64 & 0.70 & 0.70 & 0.60 & 0.49 & 0.66 & 0.74 & \textbf{0.77} & 0.69 & 0.11 & 0.09 & 0.59 & 0.68 & 0.67 & 0.10 & 0.18 \\ \hline
L\'evy(0, 0.5) & 50 & 0.05 & 0.05 & 0.05 & 0.05 & 0.05 & 0.05 & 0.05 & 0.05 & 0.05 & 0.05 & 0.05 & 0.04 & 0.05 & 0.05 & 0.05 & 0.05 & 0.05 \\
L\'evy(0, 1) & 50 & 0.05 & 0.05 & 0.05 & 0.05 & 0.05 & 0.05 & 0.05 & 0.05 & 0.05 & 0.05 & 0.05 & 0.04 & 0.05 & 0.05 & 0.05 & 0.05 & 0.05 \\
L\'evy(0, 2) & 50 & 0.05 & 0.05 & 0.05 & 0.05 & 0.05 & 0.05 & 0.04 & 0.05 & 0.05 & 0.05 & 0.05 & 0.05 & 0.05 & 0.05 & 0.05 & 0.05 & 0.05 \\
Burr(1.5, 0.5, 0.5) & 50 & 0 & 0.67 & 0.56 & 0.58 & 0.66 & \textbf{1} & 0.81 & 0.73 & 0.68 & 0.72 & 0.41 & 0.35 & \textbf{1} & \textbf{1} & \textbf{1} & 0.27 & 0.72 \\
LN(0, 1) & 50 & \textbf{1} & 0.93 & 0.99 & \textbf{1} & \textbf{1} & 0.51 & 0.83 & 0.97 & \textbf{1} & \textbf{1} & 0.62 & 0.57 & \textbf{1} & \textbf{1} & \textbf{1} & 0.74 & 0.76 \\
$\chi^2(3)$ & 50 & \textbf{1} & 0.76 & 0.94 & 0.99 & \textbf{1} & 0.36 & 0.53 & 0.82 & 0.95 & 0.99 & 0.67 & 0.72 & \textbf{1} & \textbf{1} & \textbf{1} & 0.98 & 0.95 \\
HN(0,1) & 50 & \textbf{1} & 0.48 & 0.68 & 0.89 & 0.95 & 0.55 & 0.32 & 0.46 & 0.71 & 0.91 & 0.76 & 0.73 & \textbf{1} & \textbf{1} & \textbf{1} & \textbf{1} & 0.97 \\
$\Gamma (3, 2)$ & 50 & \textbf{1} & \textbf{1} & \textbf{1} & \textbf{1} & \textbf{1} & 0.95 & \textbf{1} & \textbf{1} & \textbf{1} & \textbf{1} & 0.99 & 0.98 & \textbf{1} & \textbf{1} & \textbf{1} & \textbf{1} & 0.99 \\
W(2, 1) & 50 & \textbf{1} & \textbf{1} & \textbf{1} & \textbf{1} & \textbf{1} & 0.84 & 0.98 & \textbf{1} & \textbf{1} & \textbf{1} & 0.99 & 0.97 & \textbf{1} & \textbf{1} & \textbf{1} & \textbf{1} & \textbf{1} \\
$\Gamma(0.4, 2)$ & 50 & 0.27 & 0.59 & 0.59 & 0.58 & 0.58 & \textbf{1} & 0.74 & 0.66 & 0.57 & 0.59 & 0.06 & 0.05 & \textbf{1} & \textbf{1} & \textbf{1} & 0.65 & 0.04 \\
$W(0.4, 2)$ & 50 & 0.01 & 0.58 & 0.58 & 0.59 & 0.59 & \textbf{1} & 0.74 & 0.64 & 0.57 & 0.60 & 0.06 & 0.05 & \textbf{1} & \textbf{1} & \textbf{1} & 0.17 & 0.04 \\
$LN(0, 2)$ & 50 & 0.07 & 0.53 & 0.52 & 0.53 & 0.52 & \textbf{0.86} & 0.69 & 0.51 & 0.29 & 0.10 & 0.05 & 0.04 & 0.81 & 0.84 & 0.84 & 0.17 & 0.04 \\
$Chen(2, 0.4)$ & 50 & 0.51 & 0.64 & 0.64 & 0.64 & 0.64 & \textbf{1} & 0.80 & 0.69 & 0.58 & 0.48 & 0.08 & 0.06 & \textbf{1} & \textbf{1} & \textbf{1} & 0.73 & 0.10 \\
LG(7, 2) & 50 & \textbf{0.88} & 0.36 & 0.58 & 0.77 & 0.81 & 0.18 & 0.33 & 0.54 & 0.68 & 0.84 & 0.13 & 0.11 & 0.51 & 0.55 & 0.65 & 0.36 & 0.46 \\
LL(1, 2) & 50 & 0.43 & 0.30 & 0.13 & 0.11 & 0.20 & \textbf{0.77} & 0.48 & 0.22 & 0.09 & 0.12 & 0.06 & 0.05 & 0.72 & 0.75 & 0.76 & 0.21 & 0.19 \\
FR(1, 1) & 50 & 0.95 & 0.89 & 0.93 & 0.93 & 0.89 & 0.80 & 0.91 & 0.95 & \textbf{0.96} & 0.94 & 0.21 & 0.19 & 0.90 & 0.95 & \textbf{0.96} & 0.22 & 0.34 \\ \hline
\end{tabular}
\end{table}
\end{landscape}

{
 Results are presented in Table \ref{zdruzene_tabele}. For the sake of brevity, sample sizes are dropped from labels whenever they can be clearly determined. When the tests are compared, it can be seen that no new test is preferable among the others. That is in concordance with
 \cite{Janssen},
which asserts that the global power function of any nonparametric test is flat on balls
of alternatives except for alternatives coming from a finite-dimensional subspace.

From Table \ref{zdruzene_tabele}, it can be seen that JEL and AJEL approaches, proposed in \cite{bhati2020jackknife}, are less powerful than classical, whenever the testing is utilized via the original version of  $\vert I^{[1,1]}\vert $. It is also notable that  the power of $R_{a}$ is significantly affected by the value of tuning parameter $a$ and alternative distribution. Having all results in mind we recommend the value in the interval $[0.5,1]$. The behaviour of $J_{a}$ is less sensitive to { the} change of parameter $a$. In almost all cases both $R_{a}$ and $J_{a}$ dominate JEL and AJEL competitors. { It can be concluded that novel tests exhibit better performance  when compared with the tests $N_1^a$ and $N_1^b$, proposed in \cite{pitera2022goodness}. 
 When compared to  EDF-based tests, in some cases novel tests show better performance, while in other cases they are comparable. 
 }

 }

\section{Real data examples}\label{sec:6}
{ 
In this section, we apply the novel tests presented in this paper on two real data sets considered in \cite{bhati2020jackknife}. The data sets and their visual representations are provided in Appendix C. 

The first one { (Rainfall)} contains the weighted rainfall data for the month of January in India. Although there  is no objective
reason for modelling data with such a shape  with the L\'evy distribution, in  \cite{bhati2020jackknife}, authors concluded that the data follows the L\'evy distribution \cite[p. 10]{bhati2020jackknife}.
However, all tests we consider report $p$-values smaller than { $0.05$}, which clearly implies that L\'evy distribution is not a justified choice.

{The second data set { (Hillside)} consists of the well yields near Bel Air, Hartford county, Maryland. The $p$-values are presented in Table \ref{pvalc} { (see also Table \ref{pval_median}).} 
From the { Figure \ref{fig: well}} presented in Appendix C, it can be deduced that the empirical density { of the Hillside data} is, among the distributions studied in the simulation study, closest to {LL(1, 2).} 
Having in mind $R_{0,2}$ that
is quite powerful against this alternative, we cannot conclude that the L\'evy distribution is the appropriate model for the { Hillside data.}}


\begin{table}[htbp]\centering
\caption{$p$-values of novel tests - {} ML estimate }
\label{pvalc}
\begin{tabular}{@{}llllll@{}}
\toprule
 & $R_{0.2}$ & $R_{0.5}$ & $R_{1}$ & $R_{2}$ & $R_{5}$ \\ \midrule
Rainfall & 0.014 & 0 & 0 & 0 & 0 \\
Hillside & 0.004 & 0.006 & 0.024 & 0.106 & 0.494 \\ \midrule
 & $J_{1}$ & $J_{2}$ & $J_{5}$ & $J_{10}$ &  \\ \midrule
Rainfall & 0 & 0 & 0 & 0 &  \\
Hillside & 0.021 & 0.07 & 0.281 & 0.622 &  \\ \bottomrule
\end{tabular}
\end{table}

}
{
\section{Concluding remarks}
In this paper, we proposed {two} new goodness-of-fit tests} for the L\'evy distribution with arbitrary scale parameter {and a generalization of an existing one}. 
The asymptotic distributions of the proposed tests were derived and the local approximate Bahadur efficiencies of the proposed tests and the generalized Bhati--Kattumanil test were compared. 
Obtained empirical powers  also clearly indicate the dominance over other specific goodness-of-fit tests for the L\'evy distribution. 

{
We end our work by identifying some open research questions. From the results of our empirical study presented in Section 5 and Appendix D, we note that the test's powers are sensitive to the choice of the estimator of the scale parameter. Therefore it is of interest to further analyse  the small sample behaviour of our tests under different estimators, not included in the study. In addition, a prominent direction of further research would be to look for the adaptation of proposed tests when both location and scale parameters are unknown. This case is even more challenging for developing asymptotic properties of test statistics under the null and fixed alternative as well. It would be also interesting to look for tests' behaviour under contiguous alternatives.
}

\section*{Acknowledgements}
{We express our deep gratitude to professor Deepesh Bhati for providing us the code for computing empirical powers of JEL and AJEL test statistics.

{ We are also grateful to two anonymous referees whose suggestions have contributed to the quality of this paper.}
}
\section*{Declaration of interests}
{
The work of B. Milo\v sevi\'c is  supported by the Ministry of Science, Technological Development and Innovations of the Republic of Serbia (the contract 451-03-47/2023-01/ 200104). { The work is also supported by the COST action
CA21163 - Text, functional and other high-dimensional data in econometrics: New models, methods, applications (HiTEc).
}
\bigskip

\section*{Appendix A -- Proofs}
\begin{proof}[Proof of Theorem \ref{asimptotikaJ}]
It is easy to see that, under $H_0$, the distribution of $J_{n,a}$ does not depend on $\lambda$ which justifies the usage of the introduced class for testing the composite hypothesis. Therefore, when deriving the limiting distribution we may suppose that $\lambda=1$. 
It { should be noted} that the statistic \eqref{statJn} can be represented as 
\begin{equation*}
    J_{n, a}=\sup\limits_{t\in [0, 1]}\abs{V_n(t; a,  \hat{\lambda})},
\end{equation*}where, for each $t$ and $a$, 
\begin{equation*}\label{vstat}
    V_n(t; a, \hat\lambda)=\frac{1}{n^2}\sum\limits_{i_1, i_2}^n\Psi(X_{i_1}, X_{i_2}; t, a, \hat\lambda),
\end{equation*}
 is a $V$-statistic with the estimated parameter $\lambda$, with symmetric kernel $\Psi(\cdot; t, a)$.
 
 Applying the Taylor expansion, we get{
\begin{align}\label{razbijanje}
\sqrt{n}V_n(t; a, \hat\lambda)=\sqrt{n}V_n(t; a,1)+\sqrt{n}(\hat\lambda-1)\frac{\partial V_n(t; a,\nu)}{\partial\nu}\bigr\vert _{\nu=1}+R_n(t),\end{align}
where
\begin{align*}
    R_n(t)=\frac{\sqrt{n}}{2}(\hat\lambda-1)^2\frac{\partial^2 V_n(t; a,\nu)}{\partial\nu^2}\Bigr\vert _{\nu=\lambda_1}.
\end{align*}
and $\lambda_1$ lies between $1$ and $\hat\lambda$. We need to establish that $\sup\limits_{t\in[0,1]}|R_n(t)|$ is an $o_P(1)$ sequence.
It can be shown that 
\begin{align}\label{normalnost_lambda}
    \sqrt{n}(\hat\lambda-1)\xrightarrow{n\to\infty}\mathcal{N}(0, 2),
\end{align}
and that $\hat\lambda$ is a consistent estimator of $1$ (see \cite{ali2005inference}).
 {From (\ref{normalnost_lambda}) 
 and from Slutsky's theorem we get $\sqrt{n}(\hat\lambda-1)^2\xrightarrow{P}0$.
 
 We have that \begin{align*}
     \frac{\partial^2 \Psi(X, Y; t, a, \nu)}{\partial\nu^2}\Bigr\vert _{\nu=\lambda_1}&=\frac{t^a (-\log (t))^{5/2}}{16 \lambda_1^4} (-8 X t^{\frac{X}{\lambda_1}} (2 \lambda_1+X \log (t))-8 Y t^{\frac{Y}{\lambda_1}} (2 \lambda_1+Y \log (t))\\&+(X+Y) t^{\frac{X+Y}{4 \lambda_1}} (8 \lambda_1+\log (t) (X+Y))).
 \end{align*}
If we denote with
 \begin{equation*}
    g(t;X)=\frac{8 X t^{\frac{X}{\lambda_1}} (2 \lambda_1+X \log (t))}{16 \lambda_1^4},
    \end{equation*}
    then the following holds
 \begin{equation*}
    |g(t; X)|
    \leq\frac{1}{\lambda_1}\Big(\Big|\frac{X t^{\frac{X}{\lambda_1}}}{\lambda_1^2}\Big|+\Big|\frac{X^2 t^{\frac{X}{\lambda_1}}\log (t)}{2\lambda_1^3}\Big|\Big).
 \end{equation*} We have that
 \begin{equation*}
     \Big|\frac{X t^{\frac{X}{\lambda_1}}}{\lambda_1^2}\Big|=\frac{X t^{\frac{X}{\lambda_1}}}{\lambda_1^2}=\frac{U^2 e^{U\log(t)}}{X}, \;  U=\frac{X}{\lambda_1}
 \end{equation*}
 and noting that $\log(t)<0$,  we obtain that for every fixed $t$ function $h(u)=u^2e^{u\log(t)}$ attains its maximum at the point $u=-\frac{2}{\log(t)}$. Therefore, 
 \begin{equation}\label{korak1}
    \frac{X t^{\frac{X}{\lambda_1}}}{\lambda_1^2}\leq \frac{4e^{-2}(-\log(t))^{-2}}{X}.
 \end{equation}
 Similarly, we get that
 \begin{equation}\label{korak2}
     \Big|\frac{X^2 t^{\frac{X}{\lambda_1}}\log (t)}{2\lambda_1^3}\Big|\leq \frac{27e^{-3}(-\log(t))^{-2}}{2X}.
 \end{equation}
 Since $\lambda_1\geq \min(1, \hat\lambda)$, we get that $\frac{1}{\lambda_1}\leq\max\Big(1, \frac{1}{\hat\lambda}\Big)$. From (\ref{korak1}) and (\ref{korak2}) we conclude that
 \begin{equation}\label{ogranicenje}
     |g(t; X)| \leq \frac{(4e+13.5)\max\Big(1, \frac{1}{\hat\lambda}\Big)}{Xe^3(-\log(t))^2}.
 \end{equation}
 Function $\omega(t; a)=t^a(-\log(t))^{\frac{1}{2}}$ attains its maximum as  a continuous function of $t$ on a compact set. Applying (\ref{ogranicenje}) three times, we get
 \begin{align*}
    & \Big |\frac{\partial^2 \Psi(X, Y; t, a, \nu)}{\partial\nu^2}\Bigr\vert _{\nu=\lambda_1}\Big |= t^a (-\log (t))^{5/2}|2g\Big(t; \frac{X+Y}{4}\Big)-g(t;X)-g(t;Y))|\\
    &\leq t^a (-\log (t))^{5/2}(2\big|g\Big(t; \frac{X+Y}{4}\Big)\big|+|g(t;X)|+|g(t;Y))|\leq 
    \max\limits_{t\in[0, 1]}|t^a(-\log(t))^{\frac{1}{2}}|(4e+13.5)e^{-3}\times
    \\&\max\Big(1, \frac{1}{\hat\lambda}\Big)\Big(\frac{8}{X+Y}+\frac{1}{X}+\frac{1}{Y}\Big)=C(a)\max\Big(1, \frac{1}{\hat\lambda}\Big)\Big(\frac{8}{X+Y}+\frac{1}{X}+\frac{1}{Y}\Big).
 \end{align*}
 Note that that $C(a),
 \frac{1}{n}\sum\limits_{i}\frac{1}{X_i}$, and $\frac{1}{n^2}\sum\limits_{i, j}\frac{1}{X_i+X_j}$
 are $O_P(1)$ sequences.
 
Using the continuous mapping theorem, we get $\max\Big(1, \frac{1}{\hat\lambda}\Big)\xrightarrow{P}1$. Using the law of large numbers for $V-$ statistics, the triangle inequality and Slutsky's theorem, we  have that
\begin{align*}
    &\sup\limits_{t\in [0, 1]}\Big|\frac{\partial^2 V_n(t; a,\nu)}{\partial\nu^2}\Big|\leq \sup\limits_{t\in [0, 1]}\frac{1}{n^2}\sum\limits_{i_1, i_2}^n\Big|\frac{\partial^2 \Psi(X_{i_1}, X_{i_2}; t, a, \nu)}{\partial\nu^2}\Big|\leq
   \\&C(a)\max\Big(1, \frac{1}{\hat\lambda}\Big)\frac{1}{n^2}\sum\limits_{i_1, i_2}^n \Big(\frac{8}{X_{i_1}+X_{i_2}}+\frac{1}{X_{i_1}}+\frac{1}{X_{i_2}}\Big)\xrightarrow{P}
    C(a)E\Big(\frac{8}{X+Y}+\frac{1}{X}+\frac{1}{Y}\Big)=4C(a).
\end{align*}
We have established that $\sup\limits_{t\in [0, 1]}\Big|\frac{\partial^2 V_n(t; a,\nu)}{\partial\nu^2}\Big|$ is an $O_P(1)$ sequence. Slutsky's theorem, along with the fact that ${\sqrt{n}(\hat\lambda-1)^2}$ is an $o_P(1)$ sequence establishes that $\sup\limits_{t\in[0,1]}|R_n(t)|$ is an $o_P(1)$ sequence.
}}

 Using the law of large numbers for $V-$ statistics once more, we  have that 
\begin{align*}\frac{\partial V_n(t;a, \nu)}{\partial\nu}\Bigr\vert _{\nu=1}\overset{P}{\to}t^a(\log t)^{\frac{3}{2}}E\Big(\frac{1}{2}X_1t^{X_1}+\frac{1}{2}X_2t^{X_2}-\frac{1}{4}(X_1+X_2)t^{\frac{1}{4}(X_1+X_2)}\Big)=0.
\end{align*}
{ {The statement of the convergence  $\sup\limits_{t\in[0, 1]}\Big\vert \frac{\partial V_n(t;a, \nu)}{\partial\nu}\bigr\vert _{\nu=1}\Big\vert \overset{P}{\to} 0$, is formalized in the following lemma.}
\begin{lemma}\label{lemma1}
The limit in probability { under $H_0$} of
$\sup\limits_{t\in[0, 1]}\Bigr\vert \frac{\partial V_n(t;a, \nu)}{\partial\nu}\bigr\vert _{\nu=1}\Bigr\vert, $ as $n\to \infty$, equals 0.
\end{lemma}
\begin{proof}
Assume $a\geq 1$. Let's focus on the derivative of the kernel. The following holds:
\begin{align*}
    &\frac{\partial \Psi (X_1, X_2; t, a, \nu)}{\partial \nu}\Bigr\vert _{\nu=1}=\frac{t^a (-\log (t))^{5/2} \big(2 X_1 t^{X_1}+2 X_2 t^{X_2}-(X_1+X_2) t^{\frac{X_1+X_2}{4}}\big)}{4}.
 \end{align*}
 {
 Denote with 
 \begin{align*}
    f_n(t)=&\frac{1}{n^2}\sum\limits_{i, j}\frac{\partial \Psi (X_1, X_2; t, a, \nu)}{\partial \nu}\Bigr\vert _{\nu=1}=\frac{1}{n^2}\sum\limits_{i, j}\frac{t^a (-\log (t))^{5/2} \big(2 X_1 t^{X_1}+2 X_2 t^{X_2}-(X_1+X_2) t^{\frac{X_1+X_2}{4}}\big)}{4}.
 \end{align*}
 Function $f_n$ is continuous (as a function of $t$) and has a continuous derivative. Moreover, for every $t\in [0,1]$ we have that $Ef_n(t)=0.$}
{
Denote with 
$S_{ND}=\{t: f_n'(t)\geq 0\}$, $S_{NI}=\{t:f_n'(t)\leq 0\}$ the sets on which $f_n$ is non-decreasing and non-increasing respectively. The continuity of $f'_n(t)$ ensures that both sets are closed subsets of the compact set $[0,1]$. Therefore, $S_{ND}$ and $S_{NI}$ are compact. Note that $S_{ND}\cup S_{NI}=[0,1]$.

From the subadditivity of the supremum, we have that
\begin{align*}
\sup\limits_{t\in[0,1]}\vert  f_n(t)\vert  &\leq \sup\limits_{t\in S_{ND}}\vert  f_n(t)\vert+\sup\limits_{t\in S_{NI}}\vert  f_n(t)\vert  
\end{align*}

From the law of large numbers and continuity of the modulus, we have that for every $t\in[0, 1]$: 
\begin{align*}
|f_n(t)-Ef_n(t)|=|f_n(t)|\xrightarrow{P}0.
\end{align*}
Function $f_n(t)$ is non-decreasing for every $t\in S_{ND}$. 
By applying Lemma 1 from \cite{novoa2014testing}, we obtain \begin{align*}
     \sup\limits_{t\in S_{ND}}|f_n(t)|\xrightarrow{P}0.
 \end{align*}
 Similarly, function $f_n(t)$ is non-increasing for every $t\in S_{NI}$. 
By applying Lemma 1 from \cite{novoa2014testing}, we obtain \begin{align*}
     \sup\limits_{t\in S_{ND}}|f_n(t)|\xrightarrow{P}0.
 \end{align*}


Therefore, $\sup\limits_{t\in [0, 1]} \Big\vert \frac{\partial V_n(t;a, \nu)}{\partial\nu}\bigr\vert _{\nu=1}\Big\vert \leq  \sup\limits_{t\in S_{ND}}\vert  f_n(t)\vert+\sup\limits_{t\in S_{NI}}\vert  f_n(t)\vert   \to 0$. 
The result then follows. }
\end{proof}}

Using Lemma \ref{lemma1} and using Slutsky's theorem, we conclude that { under $H_0$} $\sqrt{n}V_n(t; a, \hat\lambda)$ and $\sqrt{n}V_n(t; a, 1)$ are asymptotically equally distributed.

The distribution of $\sqrt{n}V_n(t; a, 1)$ can be obtained from Hoeffding theorem for non-degenerate U- (V-) statistics (see, e.g., \cite{korolyuk2013theory}). 

 The first projection of the kernel $\Psi(\cdot; t,a)$ is given by
\begin{align*}
    \psi(x; t, a)&=E(\Psi(X_1, X_2; t, a\vert X_1=x))=\frac{1}{2} t^a (-\log (t))^{3/2} \Big(-2 t^{\frac{x}{4}} e^{-\frac{\sqrt{-\log (t)}}{\sqrt{2}}}+t^x+e^{-\sqrt{2} \sqrt{-\log (t)}}\Big).
\end{align*}
This function is obviously non-constant. In addition, it can be shown that  $ E\psi(X; t, a)^2<\infty$ for every $t\in [0, 1]$. 
Hence $V_n(t; a, 1)$ is non-degenerate.
Therefore,  from the Hoeffding theorem  and the multivariate central limit theorem, it follows that the finite-dimensional asymptotic distributions of $\sqrt{n} V_n(t; a, 1)$ are normal.  Hence,  it suffices to show that the sequence $\{\sqrt{n}V_n(t; a, 1)\}$ is tight. {For the sake of brevity, we { will} denote $V_n(t; a, 1)$ with $V_n(t;a)$ and we will drop the argument for $\lambda$ in the following text whenever it is equal to 1.} {The { tightness} then follows from \cite[Theorem 12.3]{billingsley2013convergence}.}

Let us denote
$$V_n(t; a)=\frac{1}{n^2}\sum\limits_{i, j}\Psi  (X_i, X_j; t, a), $$
where $\Psi $ denotes the symmetric kernel of the $V$-statistic.

To show tightness, we observe that
$$V_n(t+u;a)-V_n(t; a)=\frac{1}{n^2}\sum\limits_{i, j}\Big(\Psi (X_i, X_j; t+u,a)-\Psi (X_i, X_j; t,a)\Big).$$
Therefore
\begin{align*}
&E\Big(\sqrt{n}V_n(t+u;a)-\sqrt{n}V_n(t; a)\Big)^2\\
&=E\Big(\frac{1}{n^3}\sum\limits_{i, j, k, l}(\Psi (X_i, X_j; t+u, a)-\Psi (X_i, X_j; t, a))(\Psi (X_k, X_l; t+u, a)-\Psi (X_k, X_l; t, a))\Big).
\end{align*}
{Several different cases occur:
\begin{itemize}
    \item If the indices $i, j, k, l$ are different, the independence of the random variables $X_i, X_k, X_j, X_l$ and the characterization give us
    \begin{align*}
   E\big(&\frac{1}{n^3}\sum\limits_{i\neq j\neq k\neq l}(\Psi (X_i, X_j; t+u, a)-\Psi (X_i, X_j; t, a))(\Psi (X_k, X_l; t+u, a)-\Psi (X_k, X_l; t, a))\big)=0,
   \end{align*}
   and we have $4!\binom{n}{4}$ such cases.
    \item If three out of four indices are different, then two cases can occur:
    \begin{enumerate}
      \item If $i=j$, then the independence and characterization give us  
      \begin{align*}
          E\big(\frac{1}{n^3}\sum\limits_{i\neq k\neq l}(\Psi (X_i, X_i; t+u, a)-\Psi (X_i, X_i; t, a))
          (\Psi (X_k, X_l; t+u, a)-\Psi (X_k, X_l; t, a))\big)=0,
          \end{align*}
      and we have $8n\binom{n-1}{2}=4n(n-1)(n-2)$ such cases.
      \item If without loss of generality $i=k$, then we have that
      \begin{align*}
      &E\big(\frac{1}{n^3}\sum\limits_{i\neq j\neq l}(\Psi (X_i, X_j; t+u, a)-\Psi (X_i, X_j; t, a))
      &(\Psi (X_i, X_l; t+u, a)-\Psi (X_i, X_l; t, a))\big)\neq 0.
      \end{align*}
Since we have $8n\binom{n-1}{2}=4n(n-1)(n-2)$ such cases, the sum above reduces to
    \begin{align*}
    &\frac{4n(n-1)(n-2)}{n^3}E\big((\Psi (X_1, X_2; t+u, a)-\Psi (X_1, X_2; t, a))(\Psi (X_1, X_3; t+u, a)-\Psi (X_1, X_3; t, a))\big)\\
    &=\frac{4(n-1)(n-2)}{n^2}E(\psi(X_1; t+u, a)-\psi(X_1; t, a))^2,
    \end{align*}
    where $\psi(X; t, a)$ denotes the first projection of the kernel $\Psi $.
\end{enumerate}
    \item If two out of four indices are different, then we have three different cases.
    \begin{enumerate}
        \item If $i=k$ and $j=l$, we have that
        \begin{align*}
E\big(\frac{1}{n^3}\sum\limits_{i\neq j}(\Psi (X_i, X_j; t+u, a)-\Psi (X_i, X_j; t, a))^2\big),
        \end{align*}
        
       and since it reduces to 
     \begin{align*}
     &\frac{4}{n^3}\binom{n}{2}E(\Psi (X_1, X_2; t+u, a)-\Psi (X_1, X_2; t, a))^2=\\&\frac{2(n-1)}{n^2}E(\Psi (X_1, X_2; t+u, a)-\Psi (X_1, X_2; t, a))^2\overset{n\to\infty}{\longrightarrow}0, 
     \end{align*}
    we conclude that this part is asymptotically negligible.
   \item If $i=j$ and $k=l$, we have that
    \begin{align*}
E\big(&\frac{1}{n^3}\sum\limits_{i\neq j}(\Psi (X_i, X_i; t+u, a)-\Psi (X_i, X_i; t, a))(\Psi (X_j, X_j; t+u, a)-\Psi (X_j, X_j; t, a))\big)\\
&=\frac{2n(n-1)}{n^3}E(\Psi (X_1, X_1; t+u, a)-\Psi (X_1, X_1; t, a))^2,
    \end{align*}
    and since
    $$E(\Psi (X_1, X_1; t+u, a)-\Psi (X_1, X_1; t, a))^2<\infty, $$
    the asymptotic negligence follows.
    \item If $i=k=l\neq j$, then we have that the independence and characterization give us
    \begin{align*}
E\big(&\frac{1}{n^3}\sum\limits_{i\neq j}(\Psi (X_i, X_i; t+u, a)-\Psi (X_i, X_i; t, a))(\Psi (X_i, X_j; t+u, a)-\Psi (X_i, X_j; t, a))\big)=0.
\end{align*}
and we have $2n(n-1)$ such cases.
    \end{enumerate}
    \item The only remaining possibility is that all indices coincide. We have that 
    \begin{align*}
&E\big(\frac{1}{n^3}\sum\limits_{i}(\Psi (X_i, X_i; t+u, a)-\Psi (X_i, X_i; t, a)\big)^2=E\big(\frac{1}{n^2}(\Psi (X_i, X_i; t+u, a)-\Psi (X_i, X_i; t, a)\big)^2, 
\end{align*}
and since
    $$E(\Psi (X_1, X_1; t+u, a)-\Psi (X_1, X_1; t, a))^2<\infty, $$ 
the asymptotic negligence follows.
\end{itemize}}

We have that as $n\to\infty$ 
\begin{align*}
H(u)=&E\Big(\sqrt{n}V_n(t+u; a)-\sqrt{n}V_n(t; a)\Big)^2\to4E(\psi(X_1; t+u, a)-\psi(X_1; t, a))^2.
\end{align*}

Exploiting the mean-value theorem gives us 
$$E(\psi(X_1; t+u, a)-\psi(X_1; t, a))^2\leq E\Big (\frac{d\psi(X_1; t_1, a)}{dt}\Big)^2u^2,$$
for $t_1\in [t, t+u]$. We have that $E\Big(\frac{d\psi(X_1; t_1, a)}{dt}\Big)^2$ is continuous function
for $t\in [0, 1]$ and $a\geq 1$.
Taking into account that the set $[0, 1]$ is compact, we have that the following quantity $$C=\max\limits_{t\in [0, 1]} E\Big(\frac{d\psi(X_1; t, a)}{dt}\Big)^2$$
exists and is finite.
 { Therefore, tightness follows from $H(u)\leq Cu^2.$}


{{ We have established that $\sqrt{n}J_{n, a}\xrightarrow{D} \sup\limits_{t\in[0,1]} |\xi(t)|$, where $\{\xi(t)\}$ is  the centred Gaussian process, whose   covariance function can be } obtained from the following:
\begin{align}\label{kovarijaciona}
    K(s, t)&=\frac{s^a t^a}{4} (-\log (s))^{3/2} (-\log (t))^{3/2} E(\Psi(X, Y; t, a)\Psi(X, Z; s, a))=\frac{s^a t^a}{4} (-\log (s))^{3/2} (-\log (t))^{3/2}\times\\ &\int\limits_\mathbb{R^+}\int\limits_\mathbb{R^+}\int\limits_\mathbb{R^+}\frac{\big(-2 s^{\frac{x+z}{4}}+s^z+s^x\big) \big(-2 t^{\frac{x+y}{4}}+t^y+t^x\big) e^{\frac{1}{2} \big(-\frac{1}{x}-\frac{y+z}{y z}\big)}}{8 \sqrt{2} \pi ^{3/2} x^{3/2} y^{3/2} z^{3/2}}dxdydz.\nonumber
\end{align}
The result of the computation above is stated in the Theorem \ref{asimptotikaJ}.}
\end{proof}


\begin{proof}[Proof of Theorem \ref{asimptotikaR}]

{Analogously as before, the statistic \eqref{statRn} can be represented as 
\begin{equation*}
    R_{n, a}=V^R_n(a,  \hat{\lambda}),
\end{equation*}where, for each $t$ and $a$, 
\begin{equation*}\label{vstatR}
    V^R_n(a, \hat\lambda)=\frac{1}{n^2}\sum\limits_{i_1, i_2}^nZ(X_{i_1}, X_{i_2};  a, \hat\lambda),
\end{equation*}
 is a $V$-statistic with the estimated parameter $\lambda$, with symmetric kernel $Z(\cdot;a, \hat\lambda)$.
 Applying Taylor expansion as in (\ref{razbijanje}) and using (\ref{normalnost_lambda}), we conclude that
the normality of $\sqrt{n}V^R_n(a, \hat\lambda)$ can be obtained from the asymptotic normality of $\sqrt{n}V^R_n(a)$ applying Hoeffding theorem for non-degenerate U- (V-) statistics (see, e.g., \cite{korolyuk2013theory}).

 The first projection of the kernel $Z(\cdot;a):=Z(\cdot;a, 1)$ is given by: 
\begin{align}\label{zeta}
&\zeta(x; a)=E(Z(X_1, X_2;a)\vert X_1=x)
=-\frac{\sqrt{\pi } (3 a (a+2)+1) e^{\frac{1}{2 a}} \erf\Big(\frac{1}{\sqrt{2 a}}\Big)-\sqrt{2} \sqrt{a} (5 a+1)}{16 a^{9/2}}\\
&+\frac{ \Big(8 \pi  e^{\frac{1}{8 a+2 x}} \Big(48 a^2+24 a (x+1)+3 x (x+2)+1\Big) \erf\Big(\frac{1}{\sqrt{2( 4 a+x)}}\Big)}{(4 a+x)^{9/2}}-\frac{8\sqrt{2\pi} (20 a+5 x+1)\Big)}{(4 a+x)^4}  -\frac{3\sqrt{\pi}}{8 (a+x)^{5/2}}\nonumber.
\end{align}

\color{black}
This function is obviously non-constant. In addition, it can be shown that  $ E\zeta(X;  a)^2<\infty$ for every $t\in [0, 1]$. 
Hence $V^R_n(a, 1)$ is non-degenerate.
Analogously as before, utilizing the Hoeffding theorem  and the multivariate central limit theorem, we obtain that the limiting distribution of $\sqrt{n}V_n^R(a, \hat{\lambda})$ is normal $\mathcal{N}(0, \sigma^2_R{(a)})$, where $\sigma^2_R(a)=4E\zeta(X;  a)^2$}. 
\end{proof}
{\color{black}\begin{proof}[Proof of Theorem \ref{KL}]

It can be easily shown that the minimum of $K(\theta; \lambda)$, as a function of $\lambda$, is attained for

\begin{equation*}
    \lambda_0=\Big(\int\limits_{\mathbb{R^+}}\frac{g(x; \theta)}{x}dx\Big)^{-1}.
\end{equation*}

Therefore
\begin{equation*}
    K(\theta)=\int_{\mathbb{R^+}}\Big(\log\left(\dfrac{g(x;\theta)}{f(x;\lambda_0)}\right)g(x;\theta)\Big)dx.
\end{equation*}

\begin{align*}
    K(\theta)&=\int\limits_{\mathbb{R^+}}\log (g(x; \theta)) g(x; \theta)dx+\log\sqrt{2\pi}+\frac{1}{2}\log\Big(\int\limits_{\mathbb{R^+}}\frac{g(x; \theta)}{x}dx\Big)+\frac{1}{2}+\frac{3}{2}\int\limits_{\mathbb{R^+}}\log x g(x; \theta)dx.
\end{align*}
Under certain regularity conditions \cite{nikitinMetron}, the following holds:
\begin{equation*}
    K'(\theta)=\int\limits_{\mathbb{R^+}}\log(g(x; \theta))g'(x;\theta)dx+\dfrac{\int\limits_{\mathbb{R^+}}\frac{g'(x;\theta)}{x}dx}{2\int\limits_{\mathbb{R^+}}\frac{g(x; \theta)}{x}dx}+\frac{3}{2}\int\limits_{\mathbb{R^+}}\log x g'(x;\theta)dx.
\end{equation*}
Noting that $$g(x;0)=\sqrt{\frac{ \lambda}{2\pi}}e^{-\frac{ \lambda}{2x}}x^{-\frac{3}{2}},$$ where $\lambda>0$ is a parameter, it can be shown that $K'(0)=0.$

The second derivative of $K(\theta)$ is equal to 
\begin{align*}
  K''(\theta)&=\int\limits_{\mathbb{R^+}}\frac{(g'(x;\theta))^2}{g(x; \theta)}dx+\int\limits_{\mathbb{R^+}} \log(g(x; \theta))g''(x;\theta)dx\\&+
    \dfrac{\big(\int\limits_{\mathbb{R^+}}\frac{g''(x;\theta)}{x}dx\big)\big(\int\limits_{\mathbb{R^+}}\frac{g(x; \theta)}{x}dx\big)-\big(\int\limits_{\mathbb{R^+}}\frac{g'(x;\theta)}{x}dx\big)^2}{2\big(\int\limits_{\mathbb{R^+}}\frac{g(x; \theta)}{x}dx\big)^2}+\frac{3}{2}\int\limits_{\mathbb{R^+}}\log(x)g''(x;\theta)dx.
\end{align*}
The straightforward computation leads us to the following:
\begin{equation*}
    K''(0)=\sqrt{\frac{2\pi}{\lambda}}\int\limits_{\mathbb{R^+}} (g'(x;\theta))^2e^{-\frac{\lambda}{2x}}x^{-\frac{3}{2}}dx-\frac{\lambda}{2}\Big(\int\limits_{\mathbb{R^+}}\frac{g'(x;\theta)}{x}dx\Big)^2,
\end{equation*}
and the conclusion follows from the Maclaurin expansion of $K(\theta)$.
\end{proof}}

\begin{proof}[Proof of Theorem \ref{Tsup}.]
{We give just a broad outline of the proof.} The tail behaviour of $\sup_{t\in [0, 1]}\abs{\xi(t)}$ is equal to the inverse of supremum of its covariance function \cite{marcus1972sample}. {Having $\hat\lambda\xrightarrow{P}\lambda(\theta)$}, the law of large numbers for V-statistics for estimated parameters \cite{iverson1989effects} gives us that $V_n(t; a,  \hat{\lambda})$ converges to $A_J(\theta; t)=E_\theta(\Psi(X_{1}, X_{2}; t, a, \lambda{(\theta)}))$. 
By expanding into the Maclaurin series, we get
\begin{align*}
    A_J(\theta;t)=A_J(0;t)+\theta A_J'(\theta;t)+\frac{\theta^2}{2}A_J''(0;t)+o(\theta^2).
\end{align*}
It can be shown that $A_J(0;t)=b_J'(0)=0$. The direct calculation yields
\begin{align*}
    A_J''(0;t)=2\int_{\mathbb{R}^+}\psi(x; t, a)g'_\theta(x; 0)dx.
\end{align*}
{ Using the same arguments as in the  proof of Theorem \ref{asimptotikaJ},} it can be shown that the limit in probability of $J_{n, a}$ under the alternative equals to $\sup_{t\in[0,1]} \vert A_J(\theta;t)\vert =b_J(\theta)$, which finishes the proof.
\end{proof}
\begin{proof}[Proof of Theorem \ref{Rint}.]
The theorem can be shown analogously  as in \cite[Lemma 2.1]{milovsevic2016new}.  Hence we omit it here.
\end{proof}
{

\section*{Appendix B -- The generalization of the Bhati--Kattumanil statistic}\label{generalizacije}

One possible way to generalize the test statistic $T_n$ is to opt for the difference between U-empirical distribution functions of $\omega_1(X_1,X_2)=\frac{aX_1+bX_2}{\big(\sqrt{a}+\sqrt{b}\big)^2}$ and $\omega_2(X_1)=X_1$,  given by
\begin{align*}
   & G_n(t)=\frac{1}{n(n-1)}\sum_{i<j}{\rm I}\Big\{\frac{aX_i+bX_j}{(\sqrt{a}+\sqrt{b})^2}\leq t\Big\} \text{ and }\\
   & F_n(t)=\frac{1}{n}\sum_{i}{\rm I}\big\{X_i\leq t\big\}
\end{align*}
respectively. 
This leads us to the statistic
\begin{align}\label{testIn}
    \bar{I}^{[a,b]}_n=\int_{\mathbb{R}^+}(G_n(t)-F_n(t))dF_n(t).
    \end{align}

It is easy to show that the statistic is scale free under { the} null hypothesis. {This follows from the fact that if $X$ has the L\'evy distribution with the scale parameter $\lambda$, then $\frac{X}{\lambda}$ has the standard L\'evy distribution.} Therefore, while deriving asymptotic null properties we may assume that the sample comes from the standard L\'evy distribution.
Applying the same arguments as in \cite{bhati2020jackknife}, we come to the following statement about the limiting distribution of $ \bar{I}^{[a,b]}_n$ under $H_0$.
\begin{theorem} Let $X_1,X_2,...,X_n$ be an i.i.d. sample.
Under  $\mathcal{H}_0$ the limiting distribution of $\bar{I}^{[a,b]}_n$ is centred Gaussian, i.e. it holds
$$\sqrt{n} \bar{I}^{[a,b]}_n\xrightarrow{D}\mathcal{N}\big(0, \sigma_0^2(a,b)\big)$$
where
\begin{align}\label{dispab}
\sigma^2_0\big(a,b\big)&=Var\Big(2-P\Big(\frac{aX_1+bX_2}{(\sqrt{a}+\sqrt{b})^2}\geq X_3\vert X_1\Big)-P\Big(\frac{aX_2+bX_1}{(\sqrt{a}+\sqrt{b})^2}\geq X_3\vert X_1\Big)+P\big(X_2\leq X_1\vert X_1\big)\Big).
\end{align}
\end{theorem}
The values of $\sigma^2_{0}(a,b)$ cannot be calculated analytically. However, it is possible to calculate them numerically. Other values of $\sigma^2_{0}(a,b)$ are presented in Table \ref{tab: disperzije}.
Therefore, instead of using jackknife approach, one can also  test using  standardized  statistic
\begin{align*}
     \widetilde{I}^{[a,b]}_n=\sqrt{n}\frac{ \bar{I}^{[a,b]}_n}{\sigma_{0}(a,b)},
\end{align*}
or calculate $p$-values based on $\bar{I}^{[a,b]}_n$ using Monte Carlo approach.  Both mentioned approaches are much simpler than original proposed in \cite{bhati2020jackknife}.

\begin{table}[htbp]
\centering
\caption{Some values of $\sigma^2_{{0}}(a,b)$}
\label{tab: disperzije}
\begin{tabular}{@{}llllllllllll@{}}
\toprule
a & b  & $\sigma^2_0$ & a  & b  & $\sigma^2_0$ & a & b  & $\sigma^2_0$ & a & b  & $\sigma^2_0$ \\ \midrule
1 & 2  & 0.022621     & 3  & 10 & 0.0209729    & 2 & 5  & 0.02199      & 5 & 10 & 0.022621     \\
1 & 3  & 0.0213695    & 4  & 5  & 0.0234113    & 2 & 6  & 0.0213695    & 6 & 7  & 0.0234603    \\
1 & 4  & 0.0202296    & 4  & 6  & 0.0231973    & 2 & 7  & 0.0207807    & 6 & 8  & 0.0233495    \\
1 & 5  & 0.0192384    & 4  & 7  & 0.0229236    & 2 & 8  & 0.0202296    & 6 & 9  & 0.0231973    \\
1 & 6  & 0.0183778    & 4  & 8  & 0.022621     & 2 & 9  & 0.0197162    & 6 & 10 & 0.0230191    \\
1 & 7  & 0.0176251    & 4  & 9  & 0.0223067    & 2 & 10 & 0.0192384    & 7 & 8  & 0.0234715    \\
1 & 8  & 0.0169606    & 4  & 10 & 0.02199      & 3 & 4  & 0.0233495    & 7 & 9  & 0.0233862    \\
1 & 9  & 0.0163687    & 5  & 6  & 0.0234425    & 3 & 5  & 0.0230191    & 7 & 10 & 0.0232665    \\
1 & 10 & 0.0158373    & 5  & 7  & 0.0232926    & 3 & 6  & 0.022621     & 8 & 9  & 0.023479     \\
2 & 3  & 0.0231973    & 5  & 8  & 0.0230928    & 3 & 7  & 0.022201     & 8 & 10 & 0.0234113    \\
2 & 4  & 0.022621     & 5  & 9  & 0.0228647    & 3 & 8  & 0.0217804    & 9 & 10 & 0.0234842    \\
3 & 9  & 0.0213695    & 10 & 10 & 0.0235051    &   &    &  &   &    &  \\ \bottomrule
\end{tabular}\end{table}
Analogously to Theorem \ref{Tint}, we formulate the result in the general case.
\begin{theorem}\label{Tintgen} For an alternative $g(x;\theta)$ from
$\mathcal{G}$, the Bahadur exact slope of the statistic $I_{n}^{[a,b]}$
is 
\begin{align*}
c_{I}(\theta)=\frac{1}{\sigma_0^{2}(a,b)}\Big(\int_{\mathbb{R}^+}\varphi(x)g'_{\theta}(x; 0)dx\Big)^{2}\cdot\theta^{2}+o(\theta^{2}),\theta\to 0,
\end{align*}
where $\varphi(x)$ is the first projection of the symmetric kernel $\Phi(\cdot)$ of V-statistic that is asymptotically equivalent to $I_n^{[a,b]}$, namely
\begin{align*}
   \varphi(x)=& \Big(2-P\Big(\frac{aX_1+bX_2}{(\sqrt{a}+\sqrt{b})^2}\geq X_3\vert X_1\Big)-P\Big(\frac{aX_2+bX_1}{(\sqrt{a}+\sqrt{b})^2}\geq X_3\vert X_1\Big)+P\big(X_2\leq X_1\vert X_1\big)\Big).
\end{align*}
\end{theorem}

There is no significant difference between the statistic $I^{[1, 1]}$ and $I^{[a, b]}$ for different values of $a$ and $b$ with regard to the empirical powers against all of the alternatives mentioned in this paper and the local approximate Bahadur relative efficiencies, as can be seen in Table \ref{bahadur_I}.

\begin{table}[htbp]
\centering
\caption{Local approximate Bahadur relative efficiencies of $I^{[a, b]}$ with respect to LR test}
\label{bahadur_I}
\centering
\begin{tabular}{llllll} 
\toprule
     & $g_1^{[10]}$   & $g_2$   & $g_3^{[3]}$    & $g_4$  & $g_5$    \\ 
\hline
$I^{[1, 1]}$  & 0.59   & 0.54   & 0.73   & 0.53  & 0.41    \\
$I^{[2, 3]}$  & 0.59   & 0.54   & 0.73   & 0.53  & 0.41    \\
$I^{[5, 9]}$  & 0.58   & 0.54   & 0.73   & 0.53  & 0.41    \\
$I^{[9, 6]}$  & 0.59   & 0.54   & 0.73   & 0.53  & 0.41    \\
$I^{[10, 4]}$ & 0.57   & 0.53   & 0.72   & 0.52  & 0.40    \\
\bottomrule
\end{tabular}

\end{table}
}
\section*{Appendix C -- Real data}\label{RDATA}
In this section, the data used in Section \ref{sec:6} is given alongside with appropriate histograms. The theoretical L\'evy densities are drawn using the maximum likelihood estimate of the scale parameter $\lambda$.
\begin{figure}[ht!]
\centering
\includegraphics[width=8cm]{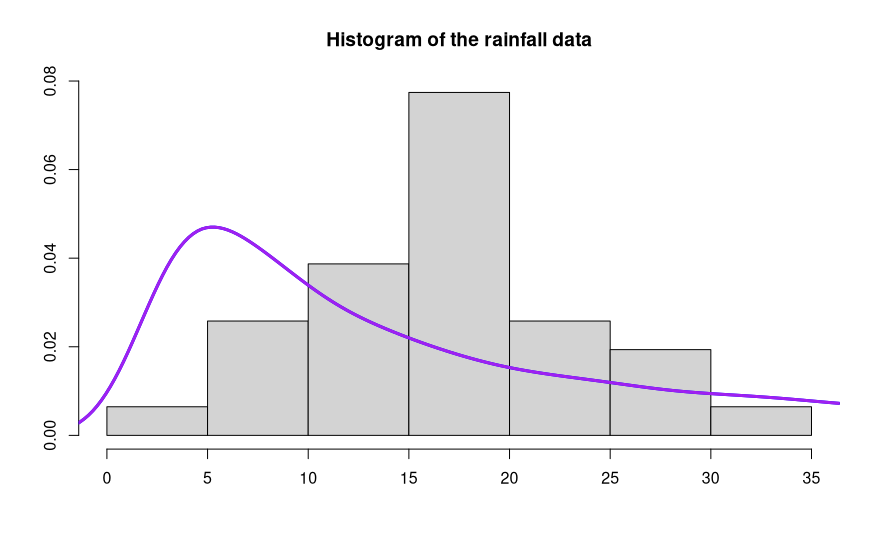}
\caption{Histogram of the data from Table \ref{kisa} and the appropriate L\'evy density. The purple line represents the L\'evy density with the scale parameter estimated by MLE ($\hat\lambda=11.82935$).
}
\label{fig: rain}
\end{figure}

\begin{table}[ht!] \centering
\caption{Weighted average of rainfall (in mm) data for India for the month of January}
\label{kisa}
\begin{tabular}{@{}llll@{}}
\toprule
Year & Rainfall & Year & Rainfall \\ \midrule
1981 & 29.3 & 1997 & 14.3 \\
1982 & 23.8 & 1998 & 16.4 \\
1983 & 18.5 & 1999 & 13.7 \\
1984 & 19 & 2000 & 18.4 \\
1985 & 23.2 & 2001 & 7.3 \\
1986 & 15.5 & 2002 & 15.7 \\
1987 & 13.2 & 2003 & 7.6 \\
1988 & 10.4 & 2004 & 25.7 \\
1989 & 15.4 & 2005 & 28.1 \\
1990 & 16 & 2006 & 17.7 \\
1991 & 14.3 & 2007 & 1.7 \\
1992 & 16 & 2008 & 18.4 \\
1993 & 18.2 & 2009 & 12 \\
1994 & 25 & 2010 & 7.5 \\
1995 & 31.3 & 2011 & 6.8 \\
1996 & 22.9 &  &  \\ \bottomrule
\end{tabular}
\end{table}

\newpage

\begin{figure}[ht!]
\centering
\includegraphics[width=8cm]{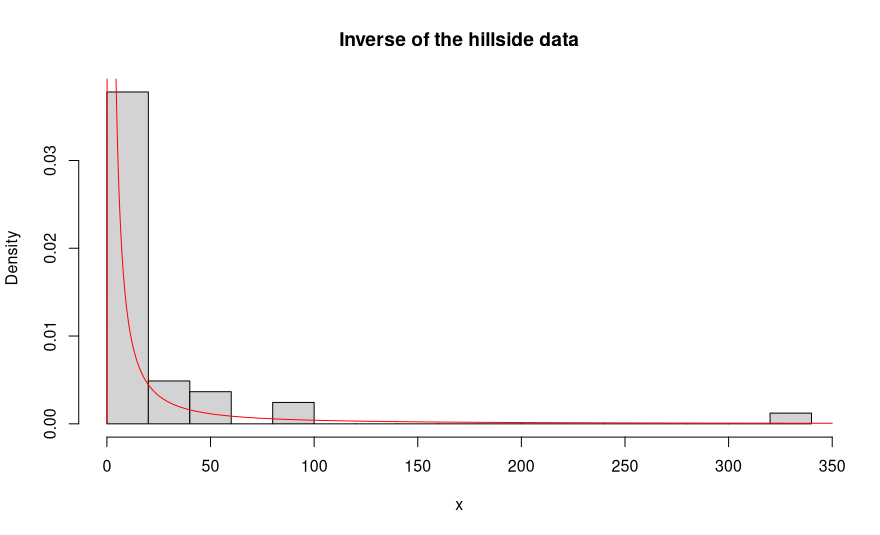}
\caption{Histogram of the inverse of the data from Table \ref{bunar} and the appropriate L\'evy density. The red line represents the L\'evy density with the scale parameter estimated by MLE ($\hat\lambda=1.052551$).
}
\label{fig: well}
\end{figure}
\begin{table}
\centering
\caption{Well yields (in gal/min/ft) based on Hillside location}
\label{bunar}
\begin{tabular}{llllllllll}
\toprule
 0.220 & 1.330 & 0.750 & 0.180 & 0.010 & 0.160 \\
 0.280 & 0.870 & 0.020 & 0.100 & 0.030 & 0.050 \\ 
 0.860 & 5.000 & 0.040 & 4.000 & 0.370 & 0.380 \\
 0.110 & 0.100 & 0.020 & 0.010 & 0.050 & 0.170 \\
 0.460 & 0.160 & 1.330 & 0.140 & 2.860 & 0.130 \\
 7.500 & 4.500 & 0.030 & 0.003 & 0.050 & 0.020 \\
 0.040 & 0.750 & 0.520 & 5.000 & 0.350 & \\
\bottomrule
\end{tabular}
\end{table}

{

\section*{Appendix D -- Median-based estimator}}\label{MBEST}
{We conduct the power study as in Section \ref{sec:5} when the median-based estimator is employed. Results are presented in Table \ref{est}.
 From Table \ref{est}, it can be seen that JEL and AJEL approaches, proposed in \cite{bhati2020jackknife}, are less powerful than classical, whenever the testing is utilized via the original version of  $\vert I^{[1,1]}\vert $. It can be concluded that novel tests are comparable with the tests $N_1^a$ and $N_1^b$ proposed in \cite{pitera2022goodness}. It is also notable that the new tests are comparable with EDF-based tests. In many cases, the new tests show better performance than the EDF-based tests. From Tables \ref{critvalJ05}, \ref{critvalJ5}, \ref{critvalL.5}, and \ref{critval5} seems that the estimation procedure doesn't significantly influence the distribution under the null hypothesis for larger sample sizes for the novel tests, which is in concordance with previously obtained theoretical results. The significant difference in test powers in Table \ref{zdruzene_tabele} and \ref{est} could be attributed to the difference in the behaviour of the estimates under the alternative distributions.

 Note that results analogous to Theorem \ref{asimptotikaJ} and \ref{asimptotikaR} could be established similarly to the MLE case. The consistency of $\hat{\lambda}_{MBE}$ will follow from \cite{mcculloch1986simple}.

The novel tests using the MBE can be applied to the real data examples from Section \ref{sec:6}. Results are presented in Table \ref{pval_median}.
 \begin{table}[htbp]
\caption{{$p$-values of novel tests - MB estimate}}
\centering
\label{pval_median}
\begin{tabular}{@{}llllll@{}}
\toprule
 & $R_{0.2}$ & $R_{0.5}$ & $R_{1}$ & $R_{2}$ & $R_{5}$ \\ \midrule
Rainfall & 0.3559 & 0.0314 & 0 & 0 & 0 \\
Hillside & 0.012 & 0.0058 & 0.042 & 0.4087 & 0.6572 \\ \midrule
 & $J_{1}$ & $J_{2}$ & $J_{5}$ & $J_{10}$ &  \\ \midrule
Rainfall & 0.026 & 0 & 0 & 0 &  \\
Hillside & 0.014 & 0.1364 & 0.7029 & 0.3086 &  \\ \bottomrule
\end{tabular}
\end{table}

From Figure \ref{fig: rain} presented in Appendix C, it can be deduced that the empirical density of the Rainfall data is, among the distributions studied in the simulation study, closest to LG(7, 2).
Since $R_{0,2}$ is the least powerful test against this alternative and all of the other tests report $p$-values smaller than 0.05, we can conclude that the L\'evy distribution is not a justified choice for the Rainfall data.

Analogously to the MLE case, $R_{0,2}$
is quite powerful against LL(1, 2). Therefore, we cannot conclude that the L\'evy distribution is the appropriate model for the Hillside data.
 }

\begin{landscape}
\begin{table}[htbp]
\caption{Comparison of empirical powers - { MB estimate}}
\label{est}
\begin{tabular}{lllllllllllllllllll}
\hline
Distribution & n & $\bar I^{[ 1, 1]}$ & $J_{1}$ & $J_{2}$ & $J_{5}$ & $J_{10}$ & $R_{0.2}$ & $R_{0.5}$ & $R_{1}$ & $R_{2}$ & $R_{5}$ & JEL & AJEL & KS & CVM & AD & $N_1^a$ & $N_1^b$ \\ \hline
L\'evy(0, 0.5) & 25 & 0.05 & 0.05 & 0.05 & 0.05 & 0.05 & 0.05 & 0.05 & 0.04 & 0.05 & 0.05 & 0.05 & 0.04 & 0.05 & 0.04 & 0.05 & 0.05 & 0.05 \\
L\'evy(0, 1) & 25 & 0.05 & 0.05 & 0.05 & 0.05 & 0.05 & 0.05 & 0.05 & 0.05 & 0.05 & 0.05 & 0.06 & 0.04 & 0.05 & 0.04 & 0.05 & 0.05 & 0.06 \\
L\'evy(0, 2) & 25 & 0.05 & 0.05 & 0.05 & 0.05 & 0.05 & 0.05 & 0.05 & 0.05 & 0.05 & 0.05 & 0.05 & 0.04 & 0.05 & 0.05 & 0.06 & 0.05 & 0.05 \\
Burr(1.5, 0.5, 0.5) & 25 & 0 & 0.99 & 0.95 & 0.84 & 0.77 & \textbf{1} & \textbf{1} & 0.99 & 0.97 & 0.87 & 0.18 & 0.14 & 0.83 & 0.86 & \textbf{1} & 0.19 & 0.44 \\
LN(0,1) & 25 & \textbf{0.99} & 0.36 & 0.83 & 0.97 & 0.99 & 0 & 0.10 & 0.70 & 0.94 & \textbf{0.99} & 0.25 & 0.21 & 0.60 & 0.78 & 0.40 & 0.30 & 0.36 \\
$\chi^2(3)$ & 25 & \textbf{1} & 0.50 & 0.94 & \textbf{1} & \textbf{1} & 0.07 & 0.13 & 0.74 & 0.98 & \textbf{1} & 0.41 & 0.36 & 0.94 & 0.95 & 0.74 & 0.58 & 0.57 \\
HN(0,1) & 25 & \textbf{1} & 0.41 & 0.91 & \textbf{1} & \textbf{1} & 0.28 & 0.10 & 0.55 & 0.95 & \textbf{1} & 0.40 & 0.35 & 0.97 & 0.96 & 0.83 & 0.73 & 0.61 \\
$\Gamma (3, 2)$ & 25 & \textbf{1} & 0.99 & \textbf{1} & \textbf{1} & \textbf{1} & 0 & 0.76 & \textbf{1} & \textbf{1} & \textbf{1} & 0.76 & 0.72 & \textbf{1} & \textbf{1} & \textbf{1} & 0.73 & 0.78 \\
W(2, 1) & 25 & \textbf{1} & 0.99 & \textbf{1} & \textbf{1} & \textbf{1} & 0 & 0.72 & \textbf{1} & \textbf{1} & \textbf{1} & 0.77 & 0.73 & \textbf{1} & \textbf{1} & \textbf{1} & 0.87 & 0.86 \\
$\Gamma(0.4, 2)$ & 25 & 0.15 & 0.89 & 0.50 & 0.15 & 0.27 & \textbf{0.98} & 0.94 & 0.71 & 0.21 & 0.12 & 0.06 & 0.05 & 0.38 & 0.40 & 0.95 & 0.21 & 0.03 \\
$W(0.4, 2)$ & 25 & 0.01 & 0.96 & 0.80 & 0.31 & 0.15 & \textbf{0.99} & 0.98 & 0.92 & 0.63 & 0.22 & 0.07 & 0.05 & 0.55 & 0.49 & 0.98 & 0.07 & 0.05 \\
$LN(0, 2)$ & 25 & 0.07 & 0.40 & 0.12 & 0.05 & 0.09 & \textbf{0.65} & 0.49 & 0.23 & 0.06 & 0.06 & 0.06 & 0.04 & 0.11 & 0.10 & 0.44 & 0.07 & 0.03 \\
$Chen(2, 0.4)$ & 25 & 0.28 & 0.77 & 0.32 & 0.19 & 0.40 & \textbf{0.94} & 0.85 & 0.48 & 0.11 & 0.22 & 0.07 & 0.05 & 0.29 & 0.34 & 0.89 & 0.25 & 0.05 \\
LG(7, 2) & 25 & \textbf{0.56} & 0.05 & 0.28 & 0.45 & 0.50 & 0.01 & 0.02 & 0.18 & 0.38 & 0.51 & 0.09 & 0.07 & 0.15 & 0.21 & 0.05 & 0.16 & 0.20 \\
LL(1, 2) & 25 & 0.22 & 0.19 & 0.06 & 0.17 & 0.24 & \textbf{0.52} & 0.29 & 0.08 & 0.07 & 0.18 & 0.06 & 0.05 & 0.06 & 0.08 & 0.33 & 0.10 & 0.09 \\
FR(1, 1) & 25 & \textbf{0.74} & 0.23 & 0.56 & 0.66 & 0.65 & 0 & 0.07 & 0.51 & 0.67 & 0.69 & 0.11 & 0.09 & 0.30 & 0.44 & 0.19 & 0.11 & 0.16 \\ \hline
L\'evy(0, 0.5) & 50 & 0.05 & 0.05 & 0.05 & 0.05 & 0.05 & 0.05 & 0.05 & 0.05 & 0.05 & 0.05 & 0.05 & 0.04 & 0.05 & 0.05 & 0.06 & 0.05 & 0.05 \\
L\'evy(0, 1) & 50 & 0.05 & 0.05 & 0.05 & 0.05 & 0.05 & 0.04 & 0.05 & 0.05 & 0.05 & 0.05 & 0.05 & 0.04 & 0.05 & 0.05 & 0.04 & 0.05 & 0.04 \\
L\'evy(0, 2) & 50 & 0.05 & 0.05 & 0.05 & 0.04 & 0.05 & 0.04 & 0.05 & 0.05 & 0.05 & 0.06 & 0.05 & 0.04 & 0.05 & 0.05 & 0.05 & 0.05 & 0.05 \\
Burr(1.5, 0.5, 0.5) & 50 & 0 & \textbf{1} & \textbf{1} & 0.99 & 0.98 & \textbf{1} & \textbf{1} & \textbf{1} & \textbf{1} & \textbf{1} & 0.29 & 0.26 & 0.99 & \textbf{1} & \textbf{1} & 0.27 & 0.72 \\
LN(0, 1) & 50 & \textbf{1} & 0.83 & 0.99 & \textbf{1} & \textbf{1} & 0 & 0.56 & 0.97 & \textbf{1} & \textbf{1} & 0.52 & 0.48 & 0.98 & \textbf{1} & 0.99 & 0.74 & 0.76 \\
$\chi^2(3)$ & 50 & \textbf{1} & 0.90 & \textbf{1} & \textbf{1} & \textbf{1} & 0.10 & 0.39 & 0.96 & \textbf{1} & \textbf{1} & 0.74 & 0.71 & \textbf{1} & \textbf{1} & \textbf{1} & 0.98 & 0.94 \\
HN(0,1) & 50 & \textbf{1} & 0.82 & \textbf{1} & \textbf{1} & \textbf{1} & 0.42 & 0.18 & 0.82 & \textbf{1} & \textbf{1} & 0.72 & 0.69 & \textbf{1} & \textbf{1} & \textbf{1} & \textbf{1} & 0.96 \\
$\Gamma (3, 2)$ & 50 & \textbf{1} & \textbf{1} & \textbf{1} & \textbf{1} & \textbf{1} & 0 & \textbf{1} & \textbf{1} & \textbf{1} & \textbf{1} & 0.99 & 0.98 & \textbf{1} & \textbf{1} & \textbf{1} & \textbf{1} & \textbf{1} \\
W(2, 1) & 50 & \textbf{1} & \textbf{1} & \textbf{1} & \textbf{1} & \textbf{1} & 0 & 0.98 & \textbf{1} & \textbf{1} & \textbf{1} & 0.99 & 0.99 & \textbf{1} & \textbf{1} & \textbf{1} & \textbf{1} & \textbf{1} \\
$\Gamma(0.4, 2)$ & 50 & 0.27 & 0.99 & 0.85 & 0.23 & 0.45 & \textbf{1} & \textbf{1} & 0.93 & 0.38 & 0.16 & 0.07 & 0.06 & 0.79 & 0.80 & \textbf{1} & 0.65 & 0.05 \\
$W(0.4, 2)$ & 50 & 0 & \textbf{1} & 0.98 & 0.62 & 0.25 & \textbf{1} & \textbf{1} & \textbf{1} & 0.92 & 0.40 & 0.07 & 0.06 & 0.91 & 0.88 & \textbf{1} & 0.18 & 0.04 \\
$LN(0, 2)$ & 50 & 0.08 & 0.67 & 0.26 & 0.05 & 0.11 & \textbf{0.90} & 0.77 & 0.45 & 0.10 & 0.06 & 0.05 & 0.04 & 0.20 & 0.22 & 0.72 & 0.16 & 0.04 \\
$Chen(2, 0.4)$ & 50 & 0.52 & 0.97 & 0.64 & 0.28 & 0.67 & \textbf{1} & 0.98 & 0.76 & 0.15 & 0.33 & 0.07 & 0.06 & 0.67 & 0.73 & \textbf{1} & 0.72 & 0.09 \\
LG(7, 2) & 50 & \textbf{0.9} & 0.17 & 0.51 & 0.73 & 0.80 & 0 & 0.07 & 0.39 & 0.67 & 0.82 & 0.16 & 0.14 & 0.32 & 0.48 & 0.29 & 0.37 & 0.46 \\
LL(1, 2) & 50 & 0.42 & 0.34 & 0.07 & 0.25 & 0.40 & \textbf{0.79} & 0.46 & 0.10 & 0.08 & 0.31 & 0.06 & 0.06 & 0.14 & 0.18 & 0.57 & 0.21 & 0.19 \\
FR(1, 1) & 50 & \textbf{0.97} & 0.68 & 0.88 & 0.92 & 0.91 & 0 & 0.55 & 0.88 & 0.94 & 0.94 & 0.21 & 0.19 & 0.63 & 0.83 & 0.74 & 0.22 & 0.35 \\ \hline
\end{tabular}
\end{table}
 \end{landscape}

{\section*{Appendix E -- Critical values of the new tests}\label{CRVNT}

In this section, the empirical 95th percentiles of the distributions of $\sqrt{n} J_{n, a}$ and $|\sqrt{n}R_{n, a}|$, under $H_0$, are presented. Different shape parameters $\lambda$ of the null distribution are used. The values are computed using Monte Carlo simulations { with} $N=100 000$ repetitions. In Tables \ref{critvalL.5} and \ref{critval5}, in  the column $\infty$, 95th percentiles of the asymptotic half-normal $HN(0, \sigma^2_R(a))$ distribution are presented.}
\begin{table}[htbp]
\centering
\caption{Critical values of $\sqrt{n}J_{n, a}$ statistic, for N=100000 repetitions and $\lambda=0.5$.}
\label{critvalJ05}
\begin{tabular}{@{}lllllllll@{}}
\toprule
n & $\sqrt{n}J_{n, 1}$ & $\sqrt{n}J_{n, 1}$ & $\sqrt{n}J_{n,2}$ & $\sqrt{n}J_{n, 2}$ & $\sqrt{n}J_{n,5}$ & $\sqrt{n}J_{n, 5}$ & $\sqrt{n}J_{n,10}$ & $\sqrt{n}J_{n, 10}$ \\
 & MLE & MED & MLE & MED & MLE & MED & MLE & MED \\ \midrule
20 & 0.14827 & 0.14335 & 0.02585 & 0.02465 & 0.00212 & 0.00215 & 0.00029 & 0.00029 \\
40 & 0.14112 & 0.13843 & 0.02511 & 0.02483 & 0.00208 & 0.00209 & 0.00028 & 0.00029 \\
60 & 0.13798 & 0.13679 & 0.02482 & 0.02450 & 0.00211 & 0.00209 & 0.00029 & 0.00029 \\
80 & 0.14025 & 0.14030 & 0.02565 & 0.02504 & 0.00211 & 0.00210 & 0.00029 & 0.00029 \\
100 & 0.13997 & 0.13954 & 0.02477 & 0.02467 & 0.00215 & 0.00214 & 0.00028 & 0.00029 \\
120 & 0.13885 & 0.13779 & 0.02482 & 0.02499 & 0.00209 & 0.00208 & 0.00029 & 0.00029 \\
140 & 0.13847 & 0.13828 & 0.02467 & 0.02455 & 0.00206 & 0.00205 & 0.00028 & 0.00029 \\
160 & 0.13808 & 0.13750 & 0.02518 & 0.02483 & 0.00210 & 0.00208 & 0.00029 & 0.00029 \\
180 & 0.14067 & 0.13904 & 0.02469 & 0.02473 & 0.00208 & 0.00210 & 0.00029 & 0.00029 \\
200 & 0.14079 & 0.13972 & 0.02519 & 0.02505 & 0.00207 & 0.00206 & 0.00028 & 0.00028 \\
220 & 0.13827 & 0.13856 & 0.02465 & 0.02444 & 0.00207 & 0.00207 & 0.00029 & 0.00029 \\
240 & 0.14066 & 0.13747 & 0.02515 & 0.02493 & 0.00208 & 0.00207 & 0.00029 & 0.00029 \\
260 & 0.13857 & 0.13848 & 0.02484 & 0.02472 & 0.00211 & 0.00211 & 0.00028 & 0.00028 \\
280 & 0.14032 & 0.14023 & 0.02509 & 0.02500 & 0.00208 & 0.00208 & 0.00029 & 0.00028 \\
300 & 0.13732 & 0.13691 & 0.02509 & 0.02475 & 0.00210 & 0.00209 & 0.00029 & 0.00029 \\
320 & 0.13790 & 0.13780 & 0.02470 & 0.02470 & 0.00208 & 0.00207 & 0.00028 & 0.00029 \\
340 & 0.13694 & 0.13777 & 0.02470 & 0.02478 & 0.00208 & 0.00208 & 0.00028 & 0.00028 \\
360 & 0.14034 & 0.13795 & 0.02491 & 0.02491 & 0.00208 & 0.00207 & 0.00028 & 0.00028 \\
380 & 0.14178 & 0.14103 & 0.02515 & 0.02497 & 0.00210 & 0.00210 & 0.00028 & 0.00028 \\
400 & 0.13703 & 0.13746 & 0.02496 & 0.02479 & 0.00209 & 0.00209 & 0.00029 & 0.00029 \\
420 & 0.13731 & 0.13813 & 0.02505 & 0.02495 & 0.00209 & 0.00209 & 0.00029 & 0.00028 \\
440 & 0.13866 & 0.13864 & 0.02472 & 0.02486 & 0.00204 & 0.00205 & 0.00029 & 0.00029 \\
460 & 0.13859 & 0.13753 & 0.02500 & 0.02494 & 0.00210 & 0.00211 & 0.00028 & 0.00028 \\
480 & 0.13863 & 0.13743 & 0.02471 & 0.02455 & 0.00207 & 0.00207 & 0.00029 & 0.00029 \\
500 & 0.13778 & 0.13719 & 0.02471 & 0.02461 & 0.00211 & 0.00210 & 0.00029 & 0.00029 \\ \bottomrule
\end{tabular}
\end{table}
\begin{table}[htbp]
\centering
\caption{Critical values of $\sqrt{n}J_{n, a}$ statistic, for N=100000 repetitions and $\lambda=5$.}
\label{critvalJ5}
\begin{tabular}{@{}lllllllll@{}}
\toprule
 n & $\sqrt{n}J_{n, 1}$ & $\sqrt{n}J_{n, 1}$ & $\sqrt{n}J_{n, 2}$ & $\sqrt{n}J_{n, 2}$ & $\sqrt{n}J_{n, 5}$ & $\sqrt{n}J_{n, 5}$ & $\sqrt{n}J_{n, 10}$ & $\sqrt{n}J_{n, 10}$ \\
& MLE & MED & MLE & MED & MLE & MED & MLE & MED \\ \midrule
20 & 0.14721 & 0.14587 & 0.02576 & 0.02487 & 0.00211 & 0.00211 & 0.00028 & 0.00028 \\
40 & 0.14223 & 0.13780 & 0.02557 & 0.02467 & 0.00213 & 0.00214 & 0.00028 & 0.00028 \\
60 & 0.14306 & 0.14023 & 0.02499 & 0.02454 & 0.00208 & 0.00207 & 0.00028 & 0.00028 \\
80 & 0.14053 & 0.13881 & 0.02519 & 0.02472 & 0.00209 & 0.00210 & 0.00029 & 0.00029 \\
100 & 0.14162 & 0.13724 & 0.02508 & 0.02492 & 0.00208 & 0.00208 & 0.00029 & 0.00029 \\
120 & 0.13977 & 0.13679 & 0.02506 & 0.02467 & 0.00207 & 0.00208 & 0.00029 & 0.00029 \\
140 & 0.13999 & 0.13811 & 0.02505 & 0.02498 & 0.00208 & 0.00207 & 0.00028 & 0.00029 \\
160 & 0.13895 & 0.13837 & 0.02482 & 0.02485 & 0.00208 & 0.00208 & 0.00029 & 0.00029 \\
180 & 0.13903 & 0.13740 & 0.02518 & 0.02474 & 0.00208 & 0.00211 & 0.00029 & 0.00028 \\
200 & 0.14000 & 0.13887 & 0.02497 & 0.02479 & 0.00208 & 0.00207 & 0.00029 & 0.00029 \\
220 & 0.13798 & 0.13963 & 0.02483 & 0.02455 & 0.00211 & 0.00213 & 0.00029 & 0.00029 \\
240 & 0.13681 & 0.13660 & 0.02480 & 0.02461 & 0.00207 & 0.00206 & 0.00029 & 0.00029 \\
260 & 0.13797 & 0.13741 & 0.02503 & 0.02470 & 0.00208 & 0.00208 & 0.00029 & 0.00029 \\
280 & 0.14082 & 0.13991 & 0.02479 & 0.02484 & 0.00206 & 0.00208 & 0.00029 & 0.00029 \\
300 & 0.13751 & 0.13719 & 0.02484 & 0.02457 & 0.00206 & 0.00206 & 0.00029 & 0.00029 \\
320 & 0.14183 & 0.13939 & 0.02544 & 0.02521 & 0.00208 & 0.00208 & 0.00028 & 0.00028 \\
340 & 0.14109 & 0.13875 & 0.02514 & 0.02509 & 0.00209 & 0.00208 & 0.00028 & 0.00028 \\
360 & 0.13746 & 0.13723 & 0.02497 & 0.02493 & 0.00209 & 0.00210 & 0.00029 & 0.00029 \\
380 & 0.14061 & 0.13883 & 0.02480 & 0.02477 & 0.00207 & 0.00206 & 0.00028 & 0.00028 \\
400 & 0.13714 & 0.13641 & 0.02415 & 0.02409 & 0.00212 & 0.00211 & 0.00029 & 0.00029 \\
420 & 0.13798 & 0.13886 & 0.02485 & 0.02459 & 0.00209 & 0.00209 & 0.00029 & 0.00029 \\
440 & 0.14116 & 0.14154 & 0.02471 & 0.02454 & 0.00210 & 0.00210 & 0.00029 & 0.00029 \\
460 & 0.14085 & 0.14010 & 0.02464 & 0.02468 & 0.00209 & 0.00209 & 0.00028 & 0.00028 \\
480 & 0.13887 & 0.13834 & 0.02513 & 0.02505 & 0.00208 & 0.00208 & 0.00029 & 0.00029 \\
500 & 0.13690 & 0.13597 & 0.02492 & 0.02454 & 0.00208 & 0.00209 & 0.00029 & 0.00029 \\ \bottomrule
\end{tabular}
\end{table}

 \begin{landscape}
\begin{table}[htbp]
\caption{Critical values of $|\sqrt{n}R_{n, a}|$ statistic, for N=100000 repetitions and $\lambda=0.5$.}
\label{critvalL.5}
\begin{tabular}{@{}lllllllllll@{}}\toprule
n & $|\sqrt{n}R_{n, 0.2}|$ & $|\sqrt{n}R_{n, 0.2}|$ & $|\sqrt{n}R_{n, 0.5}|$ & $|\sqrt{n}R_{n, 0.5}|$ & $|\sqrt{n}R_{n, 1}|$ & $|\sqrt{n}R_{n, 1}|$ & $|\sqrt{n}R_{n, 2}|$ & $|\sqrt{n}R_{n, 2}|$ & $|\sqrt{n}R_{n, 5}|$ & $|\sqrt{n}R_{n, 5}|$ \\
 & MLE & MED & MLE & MED & MLE & MED & MLE & MED & MLE & MED \\ \midrule
20 & 4.17008 & 5.90912 & 1.06865 & 1.11926 & 0.29687 & 0.27948 & 0.07024 & 0.06797 & 0.00862 & 0.00868 \\
40 & 4.19955 & 5.10804 & 1.03597 & 1.04216 & 0.28856 & 0.28013 & 0.06871 & 0.06794 & 0.00864 & 0.00865 \\
60 & 4.19686 & 4.73081 & 1.02705 & 1.02376 & 0.28602 & 0.28032 & 0.06834 & 0.06749 & 0.00865 & 0.00864 \\
80 & 4.18362 & 4.56416 & 1.03060 & 1.01520 & 0.28599 & 0.28035 & 0.06806 & 0.06746 & 0.00858 & 0.00862 \\
100 & 4.22461 & 4.45201 & 1.02478 & 1.01877 & 0.28447 & 0.28131 & 0.06825 & 0.06803 & 0.00863 & 0.00867 \\
120 & 4.18853 & 4.35032 & 1.01301 & 1.01625 & 0.28244 & 0.28059 & 0.06798 & 0.06776 & 0.00861 & 0.00861 \\
140 & 4.20833 & 4.33770 & 1.01447 & 1.01765 & 0.28174 & 0.28178 & 0.06754 & 0.06776 & 0.00855 & 0.00864 \\
160 & 4.19136 & 4.27661 & 1.02122 & 1.01196 & 0.28465 & 0.28091 & 0.06825 & 0.06772 & 0.00864 & 0.00861 \\
180 & 4.22155 & 4.27821 & 1.01443 & 1.01192 & 0.28290 & 0.28069 & 0.06798 & 0.06763 & 0.00864 & 0.00860 \\
200 & 4.19311 & 4.26874 & 1.01287 & 1.01263 & 0.28083 & 0.28204 & 0.06762 & 0.06805 & 0.00860 & 0.00859 \\
220 & 4.19972 & 4.25523 & 1.02219 & 1.01124 & 0.28281 & 0.28047 & 0.06781 & 0.06723 & 0.00860 & 0.00855 \\
240 & 4.20073 & 4.22512 & 1.01462 & 1.01582 & 0.28289 & 0.28071 & 0.06786 & 0.06759 & 0.00862 & 0.00860 \\
260 & 4.19012 & 4.24080 & 1.02344 & 1.00977 & 0.28399 & 0.28062 & 0.06824 & 0.06742 & 0.00861 & 0.00859 \\
280 & 4.21490 & 4.26364 & 1.01403 & 1.01294 & 0.28354 & 0.28093 & 0.06804 & 0.06750 & 0.00866 & 0.00857 \\
300 & 4.21452 & 4.21718 & 1.01313 & 1.01232 & 0.28188 & 0.28109 & 0.06790 & 0.06751 & 0.00860 & 0.00856 \\
320 & 4.19603 & 4.23134 & 1.01834 & 1.01366 & 0.28433 & 0.28214 & 0.06826 & 0.06802 & 0.00865 & 0.00860 \\
340 & 4.20765 & 4.23655 & 1.01740 & 1.01748 & 0.28352 & 0.28254 & 0.06788 & 0.06781 & 0.00863 & 0.00857 \\
360 & 4.20222 & 4.23949 & 1.01342 & 1.01063 & 0.28188 & 0.28022 & 0.06779 & 0.06734 & 0.00861 & 0.00859 \\
380 & 4.20538 & 4.25830 & 1.01943 & 1.00732 & 0.28328 & 0.28085 & 0.06798 & 0.06775 & 0.00864 & 0.00857 \\
400 & 4.20693 & 4.21437 & 1.01814 & 1.01072 & 0.28272 & 0.28159 & 0.06769 & 0.06774 & 0.00862 & 0.00859 \\
420 & 4.19724 & 4.20401 & 1.01645 & 1.01841 & 0.28367 & 0.28272 & 0.06801 & 0.06755 & 0.00857 & 0.00860 \\
440 & 4.16590 & 4.18322 & 1.01320 & 1.01440 & 0.28278 & 0.28204 & 0.06769 & 0.06770 & 0.00858 & 0.00861 \\
460 & 4.17603 & 4.19560 & 1.01877 & 1.00791 & 0.28326 & 0.28072 & 0.06782 & 0.06781 & 0.00858 & 0.00866 \\
480 & 4.18318 & 4.20832 & 1.00953 & 1.01192 & 0.28094 & 0.28232 & 0.06743 & 0.06808 & 0.00857 & 0.00863 \\
500 & 4.20214 & 4.20821 & 1.01907 & 1.01903 & 0.28273 & 0.28304 & 0.06776 & 0.06784 & 0.00859 & 0.00860\\
$\infty$ & 4.19818 & 4.19818 & 1.01314 & 1.01314 & 0.28191 & 0.28191 &0.06774&0.06774&0.00860 & 0.00860\\\midrule
\end{tabular}
\end{table}
\end{landscape}

\begin{landscape}
\begin{table}[htbp]
\caption{Critical values of $|\sqrt{n}R_{n, a}|$ statistic, for N=100000 repetitions and $\lambda=5$.}
\label{critval5}
\begin{tabular}{@{}lllllllllll@{}}
\toprule
n & $|\sqrt{n}R_{n, 0.2}|$ & $|\sqrt{n}R_{n, 0.2}|$ & $|\sqrt{n}R_{n, 0.5}|$ & $|\sqrt{n}R_{n, 0.5}|$ & $|\sqrt{n}R_{n, 1}|$ & $|\sqrt{n}R_{n, 1}|$ & $|\sqrt{n}R_{n, 2}|$ & $|\sqrt{n}R_{n, 2}|$ & $|\sqrt{n}R_{n, 5}|$ & $|\sqrt{n}R_{n, 5}|$ \\
 & MLE & MED & MLE & MED & MLE & MED & MLE & MED & MLE & MED \\ \midrule
20&4.15862&5.83832&1.06027&1.11188&0.29661&0.27925&0.07042&0.06805&0.00866&0.00873\\
40&4.18722&5.03037&1.04100&1.04123&0.28822&0.27920&0.06883& 0.06771 & 0.00863 & 0.00865 \\
60 & 4.19557 & 4.68703 & 1.03461 & 1.03082 & 0.28754 & 0.28193 & 0.06893 & 0.06796 & 0.00864 & 0.00866 \\
80 & 4.20230 & 4.51665 & 1.02995 & 1.01911 & 0.28610 & 0.28061 & 0.06839 & 0.06773 & 0.00864 & 0.00865 \\
100 & 4.19353 & 4.43083 & 1.02103 & 1.02076 & 0.28421 & 0.28053 & 0.06810 & 0.06766 & 0.00865 & 0.00866 \\
120 & 4.20346 & 4.35525 & 1.02439 & 1.01705 & 0.28459 & 0.28094 & 0.06815 & 0.06774 & 0.00859 & 0.00863 \\
140 & 4.17824 & 4.31838 & 1.01741 & 1.01490 & 0.28483 & 0.28138 & 0.06808 & 0.06778 & 0.00861 & 0.00862 \\
160 & 4.22484 & 4.31109 & 1.02510 & 1.01972 & 0.28409 & 0.28297 & 0.06825 & 0.06798 & 0.00856 & 0.00856 \\
180 & 4.19214 & 4.25517 & 1.01172 & 1.00914 & 0.28201 & 0.27947 & 0.06775 & 0.06744 & 0.00860 & 0.00858 \\
200 & 4.21680 & 4.28196 & 1.02439 & 1.01950 & 0.28520 & 0.28248 & 0.06827 & 0.06797 & 0.00860 & 0.00860 \\
220 & 4.20422 & 4.26522 & 1.02456 & 1.01674 & 0.28444 & 0.28273 & 0.06814 & 0.06795 & 0.00859 & 0.00861 \\
240 & 4.20907 & 4.25199 & 1.02258 & 1.01575 & 0.28307 & 0.28180 & 0.06802 & 0.06759 & 0.00861 & 0.00861 \\
260 & 4.22791 & 4.24382 & 1.01933 & 1.01740 & 0.28356 & 0.28193 & 0.06816 & 0.06807 & 0.00863 & 0.00864 \\
280 & 4.19666 & 4.20102 & 1.01446 & 1.01433 & 0.28403 & 0.28201 & 0.06813 & 0.06787 & 0.00864 & 0.00864 \\
300 & 4.20267 & 4.22986 & 1.01919 & 1.01230 & 0.28236 & 0.28172 & 0.06796 & 0.06773 & 0.00860 & 0.00860 \\
320 & 4.17988 & 4.21130 & 1.01789 & 1.01045 & 0.28301 & 0.28149 & 0.06763 & 0.06745 & 0.00859 & 0.00860 \\
340 & 4.18848 & 4.19939 & 1.01672 & 1.01201 & 0.28376 & 0.28254 & 0.06790 & 0.06776 & 0.00863 & 0.00863 \\
360 & 4.18101 & 4.21324 & 1.01395 & 1.01285 & 0.28265 & 0.28261 & 0.06794 & 0.06786 & 0.00860 & 0.00862 \\
380 & 4.21101 & 4.24101 & 1.02158 & 1.01630 & 0.28400 & 0.28305 & 0.06799 & 0.06799 & 0.00861 & 0.00860 \\
400 & 4.20429 & 4.21522 & 1.02185 & 1.01579 & 0.28337 & 0.28247 & 0.06788 & 0.06772 & 0.00862 & 0.00861 \\
420 & 4.21409 & 4.22878 & 1.01258 & 1.01327 & 0.28118 & 0.28027 & 0.06743 & 0.06736 & 0.00860 & 0.00860 \\
440 & 4.19203 & 4.19435 & 1.01509 & 1.01606 & 0.28344 & 0.28199 & 0.06772 & 0.06769 & 0.00857 & 0.00858 \\
460 & 4.19923 & 4.22726 & 1.01400 & 1.01221 & 0.28236 & 0.28124 & 0.06810 & 0.06802 & 0.00864 & 0.00865 \\
480 & 4.20922 & 4.22241 & 1.01868 & 1.01932 & 0.28370 & 0.28290 & 0.06804 & 0.06784 & 0.00858 & 0.00859 \\
500 & 4.20593 & 4.22146 & 1.01916 & 1.01679 & 0.28350 & 0.28182 & 0.06787 & 0.06770 & 0.00862 & 0.00861\\
$\infty$ & 4.19818 & 4.19818 & 1.01314 & 1.01314 & 0.28191 & 0.28191 &0.06774&0.06774&0.00860 & 0.00860\\\midrule
\end{tabular}
\end{table}
\end{landscape}


\begin{thebibliography}{}

\bibitem{achcar2018use}
Achcar, J.~A., Coelho-Barros, E.~A., Cuevas, J. R.~T., and Mazucheli, J.
  (2018).
\newblock {Use of L{\'e}vy distribution to analyze longitudinal data with
  asymmetric distribution and presence of left censored data}.
\newblock {\em Communications for Statistical Applications and Methods},
  25(1):43--60.

\bibitem{AhsNev1}
Ahsanullah, M. and Nevzorov, V.~B. (2019).
\newblock {On Some Characterizations of the Levy Distribution}.
\newblock {\em Stochastics and Quality Control}, 34(1):53--57.

\bibitem{ali2005inference}
Ali, M.~M. and Woo, J. (2005).
\newblock {Inference on reliability P (Y< X) in the Levy distribution}.
\newblock {\em Mathematical and computer modelling}, 41(8-9):965--971.

\bibitem{allison2021distribution}
Allison, J., Milo{\v{s}}evi{\'c}, B., Obradovi{\'c}, M., and Smuts, M. (2022).
\newblock Distribution-free goodness-of-fit tests for the pareto distribution
  based on a characterization.
\newblock {\em Computational Statistics}, 37(1):403--418.

\bibitem{bahadur1967rates}
Bahadur, R.~R. (1967).
\newblock {Rates of convergence of estimates and test statistics}.
\newblock {\em The Annals of Mathematical Statistics}, 38(2):303--324.

\bibitem{bhati2020jackknife}
Bhati, D. and Kattumannil, S.~K. (2020).
\newblock {Jackknife empirical likelihood test for testing one-sided L{\'e}vy
  distribution}.
\newblock {\em Journal of Applied Statistics}, 47(7):1208--1219.

\bibitem{billingsley2013convergence}
Billingsley, P. (1968).
\newblock {\em {Convergence of probability measures}}.
\newblock John Wiley \& Sons, New York.

\bibitem{SORT}
Cupari{\'c}, M., Milo{\v{s}}evi{\'c}, B., and Obradovi{\'c}, M. (2019).
\newblock {New ${L}^{2}$-type exponentiality tests}.
\newblock {\em SORT}, 43(1):25--50.

\bibitem{nasArxiv}
Cupari{\'c}, M., Milo{\v{s}}evi{\'c}, B., and Obradovi{\'c}, M. (2022).
\newblock {New consistent exponentiality tests based on {V}-empirical Laplace
  transforms with comparison of efficiencies}.
\newblock {\em Revista de la Real Academia de Ciencias Exactas, F\' isicas y
  Naturales. Serie A. Matem\'aticas}, 116:42(42):1--26.

\bibitem{Ebeling}
Ebeling, W., Romanovsky, M., and Sokolov, I. (2009).
\newblock {Velocity Distributions and Kinetic Equations for Plasmas Including
  Levy Type Power Law Tails}.
\newblock {\em Contributions to Plasma Physics}, 49:704 -- 712.

\bibitem{ebner2023test}
Ebner, B. (2023).
\newblock The test of exponentiality based on the mean residual life function
  revisited.
\newblock {\em Journal of Nonparametric Statistics}, pages 1--21.

\bibitem{feller2008introduction}
Feller, W. (2008).
\newblock {\em {An introduction to probability theory and its applications}},
  volume~2.
\newblock John Wiley \& Sons, New York.

\bibitem{iverson1989effects}
Iverson, H. and Randles, R. (1989).
\newblock {The effects on convergence of substituting parameter estimates into
  U-statistics and other families of statistics}.
\newblock {\em Probability Theory and Related Fields}, 81(3):453--471.

\bibitem{Janssen}
Janssen, A. (2000).
\newblock {Global power functions of goodness of fit tests}.
\newblock {\em The Annals of Statistics}, 28(1):239--253.

\bibitem{korolyuk2013theory}
Korolyuk, V.~S. and Borovskich, Y.~V. (2013).
\newblock {\em {Theory of U-statistics}}, volume 273.
\newblock Springer Science \& Business Media, Dodrecht.

\bibitem{koutrouvelis1980regression}
Koutrouvelis, I.~A. (1980).
\newblock Regression-type estimation of the parameters of stable laws.
\newblock {\em Journal of the American Statistical Association},
  75(372):918–928.

\bibitem{ley2009cam}
Ley, C. and Paindaveine, D. (2009).
\newblock {Le Cam optimal tests for symmetry against Ferreira and Steel's
  general skewed distributions}.
\newblock {\em Journal of Nonparametric Statistics}, 21(8):943--967.

\bibitem{lilliefors1967kolmogorov}
Lilliefors, H.~W. (1967).
\newblock {On the Kolmogorov-Smirnov test for normality with mean and variance
  unknown}.
\newblock {\em Journal of the American statistical Association},
  62(318):399--402.

\bibitem{marcus1972sample}
Marcus, M.~B. and Shepp, L.~A. (1972).
\newblock {Sample behavior of Gaussian processes}.
\newblock In {\em Proc. of the Sixth Berkeley Symposium on Math. Statist. and
  Prob}, volume~2, pages 423--421.

\bibitem{mcculloch1986simple}
McCulloch, J.~H. (1986).
\newblock {Simple consistent estimators of stable distribution parameters}.
\newblock {\em Communications in statistics-simulation and computation},
  15(4):1109--1136.

\bibitem{meintanis2022bahadur}
Meintanis, S., Milo{\v{s}}evi{\'c}, B., and Obradovi{\'c}, M. (2022).
\newblock Bahadur efficiency for certain goodness-of-fit tests based on the
  empirical characteristic function.
\newblock {\em Metrika}, pages 1--29.

\bibitem{milovsevic2016asymptotic}
Milo{\v{s}}evi{\'c}, B. (2016).
\newblock {Asymptotic efficiency of new exponentiality tests based on a
  characterization}.
\newblock {\em Metrika}, 79(2):221--236.

\bibitem{milovsevic2016new}
Milo{\v{s}}evi{\'c}, B. and Obradovi{\'c}, M. (2016a).
\newblock {New class of exponentiality tests based on U-empirical Laplace
  transform}.
\newblock {\em Statistical Papers}, 57(4):977--990.

\bibitem{PUBL}
Milo{\v{s}}evi{\'c}, B. and Obradovi{\'c}, M. (2016b).
\newblock {Some characterization based exponentiality tests and their Bahadur
  efficiencies}.
\newblock {\em Publications de l'Institut Mathematique}, 100(114):107--117.

\bibitem{nikitinKnjiga}
Nikitin, {\relax Ya}.~{\relax Yu}. (1995).
\newblock {\em {Asymptotic efficiency of nonparametric tests}}.
\newblock Cambridge University Press, New York.

\bibitem{nikitinMetron}
Nikitin, {\relax Ya}.~{\relax Yu}. and Peaucelle, I. (2004).
\newblock {Efficiency and local optimality of nonparametric tests based on
  U-and V-statistics}.
\newblock {\em Metron}, 62(2):185--200.

\bibitem{NIKITINLOGISTICKA}
Nikitin, {\relax Ya}.~{\relax Yu}. and Ragozin, I. (2020).
\newblock {Goodness-of-fit tests for the logistic address family}.
\newblock {\em Journal of Applied Statistics}, 47(13-15):2610--2622.

\bibitem{nolan2001maximum}
Nolan, J.~P. (2001).
\newblock {Maximum likelihood estimation and diagnostics for stable
  distributions}.

\bibitem{novoa2014testing}
Novoa-Mu{\~n}oz, F. and Jim{\'e}nez-Gamero, M. (2014).
\newblock {Testing for the bivariate Poisson distribution}.
\newblock {\em Metrika}, 77(6):771--793.

\bibitem{obradovic2015goodness}
Obradovi{\'c}, M., Jovanovi{\'c}, M., and Milo{\v{s}}evi{\'c}, B. (2015).
\newblock {Goodness-of-fit tests for Pareto distribution based on a
  characterization and their asymptotics}.
\newblock {\em Statistics}, 49(5):1026--1041.

\bibitem{o1998note}
O'Reilly, F. and Rueda, R. (1998).
\newblock {A note on the fit for the Levy distribution}.
\newblock {\em Communications in Statistics-Theory and Methods},
  27(7):1811--1821.

\bibitem{pitera2022goodness}
Pitera, M., Chechkin, A., and Wy{\l}oma{\'n}ska, A. (2022).
\newblock {Goodness-of-fit test for $\alpha$-stable distribution based on the
  quantile conditional variance statistics}.
\newblock {\em Statistical Methods \& Applications}, 31(2):387--424.

\bibitem{raghavachari1970theorem}
Raghavachari, M. (1970).
\newblock {On a theorem of Bahadur on the rate of convergence of test
  statistics}.
\newblock {\em The Annals of Mathematical Statistics}, pages 1695--1699.

\bibitem{ragozin2021new}
Ragozin, I.~A. (2021).
\newblock {New goodness-of-fit tests for the family of Rayleigh distributions,
  based on a special property and a characterization}.
\newblock {\em Zapiski Nauchnykh Seminarov POMI}, 505:230--243.

\bibitem{rogers2008multiple}
Rogers, G.~L. (2008).
\newblock {Multiple path analysis of reflectance from turbid media}.
\newblock {\em JOSA A}, 25(11):2879--2883.

\bibitem{tian2016parameter}
Tian, G. (2016).
\newblock {\em {Parameter Estimation for Stable Distribution: Spacing based and
  Indirect Inference}}.
\newblock PhD thesis, UC Santa Barbara.

\bibitem{vinaya2018effect}
Vinaya, M. and Ignatius, R.~P. (2018).
\newblock {Effect of L{\'e}vy noise on the networks of Izhikevich neurons}.
\newblock {\em Nonlinear Dynamics}, 94(2):1133--1150.

\bibitem{West}
West, B.~J., Allegrini, P., and Grigolini, P. (1998).
\newblock {Dynamical Generators of L{\'{e}}vy Statistics in Biology}.
\newblock In {\em Fractals in Biology and Medicine}, pages 17--29. Birkhäuser
  Basel.

\bibitem{Zol}
Zolotarev, V.~M. (1986).
\newblock {\em {One-dimensional stable distributions}}, volume~65.
\newblock American Mathematical Society, Providence, RI.

\end{thebibliography}
\end{document}